\documentclass[a4paper,12pt,twoside]{article}
\usepackage{amsmath, amsfonts, amssymb, amscd}
\usepackage{url}
\usepackage{color}
\usepackage{graphicx}
\usepackage{wrapfig}
\usepackage{mathrsfs}
\usepackage[pdfborder={0 0 0}, colorlinks=true, linkcolor=black, citecolor=black, urlcolor=gray, breaklinks=false]{hyperref}
\definecolor{gray}{rgb}{.17,.21,.21}

\newcommand{\flow}{\varphi}
\newcommand{\iso}{{\rm Iso}}

\newtheorem{theorem}{Theorem}[section]
\newtheorem{proposition}[theorem]{Proposition}
\newtheorem{lemma}[theorem]{Lemma}
\newtheorem{corollary}[theorem]{Corollary}
\newtheorem{definition}[theorem]{Definition}

\newenvironment{proof}{\medskip \noindent {\em Proof:}}{\hfill $\square$ \\[2mm] \indent}

\numberwithin{equation}{section}

\newlength{\dinwidth}
\newlength{\dinmargin}

\usepackage{oldgerm}
\usepackage{color}

\newcommand{\Ibb}[1]{ {\rm I\ifmmode\mkern -3.6mu\else\kern -.2em\fi#1}}
\newcommand{\ibb}[1]{\leavevmode\hbox{\kern.3em\vrule
     height 1.2ex depth -.3ex width .2pt\kern-.3em\rm#1}}
\newcommand{\Cl}{{\ibb C}}
\newcommand{\Rl}{{\Ibb R}}

\definecolor{lightgray}{rgb}{0.8,0.8,0.8}

\newcommand{\Om}{\Omega}
\newcommand{\om}{\omega}

\newcommand{\la}{\lambda}

\newcommand{\eps}{\varepsilon}

\newcommand{\A}{\mathcal{A}}

\newcommand{\B}{\mathcal{B}}
\newcommand{\C}{\mathcal{C}}

\newcommand{\K}{\mathcal{K}}

\newcommand{\Hil}{\mathcal{H}}

\newcommand{\E}{\mathcal{E}}

\newcommand{\DD}{\mathcal{D}}
\newcommand{\W}{\mathcal{W}}

\newcommand{\hti}{\tilde{h}}

\newcommand{\Wti}{\tilde{W}}

\newcommand{\pti}{\tilde{p}}
\newcommand{\xiti}{\tilde{\xi}}

\newcommand{\fF}{\frak{F}}

\newcommand{\frA}{\frak A}

\def\bs{{\mbox{\boldmath{$s$}}}}

\newcommand{\Aut}{\text{Aut}}
\newcommand{\SO}{\text{SO}}
\newcommand{\OO}{O}
\newcommand{\GL}{\mathrm{GL}}
\newcommand{\SL}{\text{SL}}

\newcommand{\CAR}{\mathrm{CAR}\,}

\newcommand{\supp}{\text{supp}\,}

\begin{document} 

\par 
\bigskip 
\LARGE 
\noindent
{\bf Deformations of quantum field theories\\ on spacetimes with Killing vector fields}
\bigskip 
\par 
\rm 
\large
\noindent 
{\bf Claudio Dappiaggi$^{1,a}$}, {\bf Gandalf Lechner$^{2,b}$}, {\bf Eric Morfa-Morales$^{3,c}$} \\

\par
\small
\noindent $^1$  
II. Institut f\"ur Theoretische Physik,
D-22763 Hamburg, Deutschland

\vskip .2cm

\noindent $^2$ 
Faculty of Physics, University of Vienna, A-1090 Vienna, Austria\smallskip

\vskip .2cm

\noindent $^3$
Erwin Schr\"odinger Institute for Mathematical Physics, Vienna, A-1090 Vienna, Austria\smallskip

\vskip .2cm

\noindent $^a$claudio.dappiaggi@esi.ac.at,
$^b$gandalf.lechner@univie.ac.at, $^c$emorfamo@esi.ac.at\\ 
 \normalsize

\par

\par 
\bigskip 

\noindent 
\small 
{\bf Abstract}.
The recent construction and analysis of deformations of quantum field theories by warped convolutions is extended to a class of curved spacetimes. These spacetimes carry a family of wedge-like regions which share the essential causal properties of the Poincar\'e transforms of the Rindler wedge in Minkowski space. In the setting of deformed quantum field theories, they play the role of typical localization regions of quantum fields and observables. As a concrete example of such a procedure, the deformation of the free Dirac field is studied.

\normalsize
\bigskip 

\normalsize
\bigskip


\section{Introduction}

Deformations of quantum field theories arise in different contexts and have been studied from different points of view in recent years. One motivation for considering such models is a possible noncommutative structure of spacetime at small scales, as suggested by combining classical gravity and the uncertainty principle of quantum physics \cite{DoplicherFredenhagenRoberts:1995}. Quantum field theories on such noncommutative spaces can then be seen as deformations of usual quantum field theories, and it is hoped that they might capture some aspects of a still elusive theory of quantum gravity ({\em cf.} \cite{Szabo:2003} for a review). By now there exist several different types of deformed quantum field theories, see \cite{GrosseWulkenhaar:2005,BalachandranPinzulQureshiVaidya:2007,Soloviev:2008,GrosseLechner:2008,BlaschkeGieresKronbergerSchwedaWohlgenannt:2008,BahnsDoplicherFredenhagenPiacitelli:2010} for some recent papers, and references cited therein.

Certain deformation techniques arising from such considerations can also be used as a device for the construction of new models in the framework of usual quantum field theory on commutative spaces  \cite{GrosseLechner:2007,BuchholzSummers:2008,GrosseLechner:2008,BuchholzLechnerSummers:2010,LongoWitten:2010}, independent of their connection to the idea of noncommutative spaces. From this point of view, the deformation parameter plays the role of a coupling constant which changes the interaction of the model under consideration, but leaves the classical structure of spacetime unchanged.

Deformations designed for either describing noncommutative spacetimes or for constructing new models on ordinary spacetimes have been studied mostly in the case of a flat manifold, either with a Euclidean or Lorentzian signature. In fact, many approaches rely on a preferred choice of Cartesian coordinates in their very formulation, and do not generalize directly to curved spacetimes. The analysis of the interplay between spacetime curvature and deformations involving noncommutative structures thus presents a challenging problem. As a first step in this direction, we study in the present paper how certain deformed quantum field theories can be formulated in the presence of external gravitational fields ({\em i.e.}, on curved spacetime manifolds), see also \cite{AschieriBlohmannDimitrijevicMeyerSchuppWess:2005,OhlSchenkel:2009} for other approaches to this question. We will not address here the fundamental question of dynamically coupling the matter fields with a possible noncommutative geometry of spacetime  \cite{PaschkeVerch:2004, Steinacker:2007}, but rather consider as an intermediate step deformed quantum field theories on a fixed Lorentzian background manifold $M$.
\\\\
A deformation technique which is well suited for our purposes is that of warped convolutions, see \cite{BuchholzSummers:2008} and \cite{GrosseLechner:2007,GrosseLechner:2008} for precursors and related work. Starting from a Hilbert space $\Hil$ carrying a representation $U$ of $\Rl^n$, the warped convolution $A_Q$ of an operator $A$ on $\Hil$ is defined as
\begin{align}\label{warp}
 A_Q = (2\pi)^{-n}\int d^nx\,d^ny\,e^{-ixy}\,U(Qx)AU(y-Qx)\,.
\end{align}
Here $Q$ is an antisymmetric $(n\times n)$-matrix playing the role of deformation parameter, and the integral can be defined in an oscillatory sense if $A$ and $U$ meet certain regularity requirements. For deformations of a single algebra, the mapping $A\to A_Q$ has many features in common with deformation quantization and the Weyl-Moyal product, and in fact was recently shown \cite{BuchholzLechnerSummers:2010} to be equivalent to specific representations of Rieffel's deformed $C^*$-algebras with $\Rl^n$-action \cite{Rieffel:1992}. In application to field theory models, however, one has to deform a whole family of algebras, corresponding to subsystems localized in spacetime, and the parameter $Q$ has to be replaced by a family of matrices $\{Q\}$ adapted to the geometry of the underlying spacetime.

To apply this scheme to quantum field theories on curved manifolds, we will consider spacetimes with a sufficiently large isometry group containing two commuting Killing fields, which give rise to a representation of $\Rl^2$ as required in \eqref{warp}. This setting is wide enough to encompass a number of cosmologically relevant manifolds such as Friedmann-Robertson-Walker spacetimes, or Bianchi models. Making use of the algebraic framework of quantum field theory \cite{Haag:1996, Araki:1999}, we can then formulate quantum field theories in an operator-algebraic language and study their deformations. Despite the fact that the warped convolution was invented for the deformation of Minkowski space quantum field theories, it turns out that all reference to the particular structure of flat spacetime, such as Poincar\'e transformations and a Poincar\'e invariant vacuum state, can be avoided.

We are interested in understanding to what extent the familiar structure of quantum field theories on curved spacetimes is preserved under such deformations, and investigate in particular covariance and localization properties. Concerning locality, it is known that in warped models on Minkowski space, point-like localization is weakened to localization in certain infinitely extended, wedge-shaped regions \cite{GrosseLechner:2007,BuchholzSummers:2008,GrosseLechner:2008,BuchholzLechnerSummers:2010}. These regions are defined as Poincar\'e transforms of the Rindler wedge
\begin{align}
 W_R := \{(x_0,x_1,x_2,x_3)\in \Rl^4 \,:\, x_1 > |x_0| \}
\,.
\end{align}
Because of their intimate relation to the Poincar\'e symmetry of Minkowski spacetime, it is not obvious what a good replacement for such a collection of regions is in the presence of non-vanishing curvature. In fact, different definitions are possible, and wedges on special manifolds have been studied by many authors in the literature \cite{Kay:1985,BorchersBuchholz:1999,Rehren:2000,BuchholzMundSummers:2001,GuidoLongoRobertsVerch:2001,BuchholzSummers:2004-2,LauridsenRibeiro:2007,Strich:2008,Borchers:2009}.

In Section \ref{sec:geometry}, the first main part of our investigation, we show that on those four-dimensional curved spacetimes which allow for the application of the deformation methods in \cite{BuchholzLechnerSummers:2010}, and thus carry two commuting Killing fields, there also exists a family of wedges with causal properties analogous to the Minkowski space wedges. Because of the prominent role wedges play in many areas of Minkowski space quantum field theory \cite{BisognanoWichmann:1975, Borchers:1992, Borchers:2000, BuchholzDreyerFlorigSummers:2000, BrunettiGuidoLongo:2002}, this geometric and manifestly covariant construction is also of interest independently of its relation to deformations.

In Section \ref{sec:deformation}, we then consider quantum field theories on curved spacetimes, and deform them by warped convolution. We first show that these deformations can be carried through in a model-independent, operator-algebraic framework, and that the emerging models share many structural properties with deformations of field theories on flat spacetime (Section \ref{sec:generaldeformations}). In particular, deformed quantum fields are localized in the wedges of the considered spacetime. These and further aspects of deformed quantum field theories are also discussed in the concrete example of a Dirac field in Section \ref{sec:dirac}. Section \ref{sec:outlook} contains our conclusions.


\section{Geometric setup}\label{sec:geometry}

To prepare the ground for our discussion of deformations of quantum field theories on curved backgrounds, we introduce in this section a suitable class of spacetimes and study their geometrical properties. In particular, we show how the concept of {\em wedges}, known from Minkowski space, generalizes to these manifolds. Recall in preparation that a wedge in four-dimensional Minkowski space is a region which is bounded by two non-parallel characteristic hyperplanes \cite{ThomasWichmann:1997}, or, equivalently, a region which is a connected component of the causal complement of a two-dimensional spacelike plane. The latter definition has a natural analogue in the curved setting. Making use of this observation, we construct corresponding wedge regions in Section \ref{sec:edges+wedges}, and analyse their covariance, causality and inclusion properties. At the end of that section, we compare our notion of wedges to other definitions which have been made in the literature \cite{BorchersBuchholz:1999, BuchholzMundSummers:2001, BuchholzSummers:2004-2, LauridsenRibeiro:2007, Borchers:2009}, and point out the similarities and differences.

In Section \ref{sec:examples}, the abstract analysis of wedge regions is complemented by a number of concrete examples of spacetimes fulfilling our assumptions.

\subsection{Edges and wedges in curved spacetimes}\label{sec:edges+wedges}

In the following, a spacetime $(M,g)$ is understood to be a four-dimensional, Hausdorff, (arcwise) connected, smooth manifold $M$ endowed with a smooth, Lorentzian metric $g$ whose signature is $(+,-,-,-)$. Notice that it is automatically guaranteed that $M$ is also paracompact and second countable \cite{Geroch:1968,Geroch:1970}. The (open) {\em causal complement} of a set $\OO\subset M$ is defined as
\begin{align}
 \OO ' := M \backslash \Big[\overline{J^+(\OO)}\cup \overline{J^-(\OO)}\Big]\,,
\end{align}
where $J^\pm(\OO)$ is the causal future respectively past of $\OO$ in $M$ \cite[Section 8.1]{Wald:1984}.

To avoid pathological geometric situations such as closed causal curves, and also to define a full-fledged Cauchy problem for a free field theory whose dynamics is determined by a second order hyperbolic partial differential equation, we will restrict ourselves to globally hyperbolic spacetimes. So in particular, $M$ is orientable and time-orientable, and we fix both orientations. While this setting is standard in quantum field theory on curved backgrounds, we will make additional assumptions regarding the structure of the isometry group $\iso(M,g)$ of $(M,g)$, motivated by our desire to define wedges in $M$ which resemble those in Minkowski space.

Our most important assumption on the structure of $(M,g)$ is that it admits two linearly independent, spacelike, complete, commuting smooth Killing fields $\xi_1,\xi_2$, which will later be essential in the context of deformed quantum field theories. We refer here and in the following always to pointwise linear independence, which entails in particular that these vector fields have no zeros. Denoting the flows of $\xi_1,\xi_2$ by $\flow_{\xi_1},\flow_{\xi_2}$, the orbit of a point $p\in M$ is a smooth two-dimensional spacelike embedded submanifold of $M$,
\begin{align}\label{edge1}
 E := \{\flow_{\xi_1,s_1}(\flow_{\xi_2,s_2}(p))\in M \,:\, s_1,s_2\in\Rl\}\,,
\end{align}
where $s_1,s_2$ are the flow parameters of $\xi_1,\xi_2$.

Since $M$ is globally hyperbolic, it is isometric to a smooth product manifold $\Rl\times\Sigma$, where $\Sigma$ is a smooth, three-dimensional embedded Cauchy hypersurface. It is known that the metric splits according to $g=\beta d\mathcal{T}^2- h$ with a temporal function $\mathcal{T}:\Rl\times\Sigma\to\Rl$ and a positive function $\beta\in C^\infty(\Rl\times\Sigma,(0,\infty))$, while $h$ induces a smooth Riemannian metric on $\Sigma$ \cite[Thm. 2.1]{BernalSanchez:2005}.  We assume that, with $E$ as in \eqref{edge1}, the Cauchy surface $\Sigma$ is smoothly homeomorphic to a product manifold $I\times E$, where $I$ is an open interval or the full real line. Thus $M\cong \Rl\times I\times E$, and we require in addition that there exists a smooth embedding $\iota:\Rl\times I\to M$. By our assumption on the topology of $I$, it follows that $(\Rl\times I,\iota^*g)$ is a globally hyperbolic spacetime without null focal points, a feature that we will need in the subsequent construction of wedge regions.

\begin{definition}\label{admissible}
A spacetime $(M,g)$ is called admissible if it admits two linearly independent, spacelike, complete, commuting, smooth Killing fields $\xi_1,\xi_2$ and the corresponding split $M\cong \Rl\times I\times E$, with $E$ defined in \eqref{edge1}, has the properties described above.

The set of all ordered pairs $\xi:=(\xi_1,\xi_2)$ satisfying these conditions for a given admissible spacetime $(M,g)$ is denoted $\Xi(M,g)$. The elements of $\Xi(M,g)$ will be referred to as Killing pairs.
\end{definition}

For the remainder of this section, we will work with an arbitrary but fixed admissible spacetime $(M,g)$, and usually drop the $(M,g)$-dependence of various objects in our notation, {\it e.g.}, write $\Xi$ instead of $\Xi(M,g)$ for the set of Killing pairs, and $\iso$ in place of $\iso(M,g)$ for the isometry group. Concrete examples of admissible spacetimes, such as Friedmann-Robertson-Walker-, Kasner- and Bianchi-spacetimes, will be discussed in Section \ref{sec:examples}.

The flow of a Killing pair $\xi\in\Xi$ is written as
\begin{align}
 \flow_{\xi,s}:=\flow_{\xi_1,s_1}\circ \flow_{\xi_2,s_2}=\flow_{\xi_2,s_2}\circ \flow_{\xi_1,s_1}\,,\qquad \xi=(\xi_1,\xi_2)\in\Xi,\;\;s=(s_1,s_2)\in\Rl^2,
\end{align}
where $s_1,s_2\in\Rl$ are the parameters of the (complete) flows $\flow_{\xi_1}, \flow_{\xi_2}$ of $\xi_1,\xi_2$. By construction, each $\flow_\xi$ is an isometric $\Rl^2$-action by diffeomorphisms on $(M,g)$, {\it i.e.}, $\varphi_{\xi,s}\in \iso$ and $\varphi_{\xi,s}\varphi_{\xi,u}=\varphi_{\xi,s+u}$ for all $s,u\in\Rl^2$.

On the set $\Xi$, the isometry group $\iso$ and the general linear group $\GL(2,\Rl)$ act in a natural manner.
\begin{lemma}\label{lemma:group-xi}
 Let $h\in \iso$, $N\in \GL(2,\Rl)$, and define, $\xi=(\xi_1,\xi_2)\in\Xi$,
\begin{align}
 h_*\xi &:= (h_*\xi_1,\,h_*\xi_2)\,,\\
 (N\xi)(p) &:= N(\xi_1(p),\xi_2(p))\,,\qquad p\in M\,.
\end{align}
These operations are commuting group actions of \,$\iso$ and $\GL(2,\Rl)$ on $\Xi$, respectively. The $\GL(2,\Rl)$-action transforms the flow of $\xi\in\Xi$ according to, $s\in\Rl^2$,
\begin{align}\label{N-flow}
 \varphi_{N\xi,s} &= \varphi_{\xi,N^Ts}\,.
\end{align}
If $h_*\xi=N\xi$ for some $\xi\in\Xi$, $h\in\iso$, $N\in\GL(2,\Rl)$, then $\det N=\pm1$.
\end{lemma}
\begin{proof}
Due to the standard properties of isometries, $\iso$ acts on the Lie algebra of Killing fields by the push-forward  isomorphisms $\xi_1\mapsto h_*\xi_1$ \cite{O'Neill:1983}. Therefore, for any $(\xi_1,\xi_2)\in\Xi$, also the vector fields $h_*\xi_1, h_*\xi_2$ are spacelike, complete, commuting, linearly independent, smooth Killing fields. The demanded properties of the splitting $M\cong \Rl\times I\times E$ directly carry over to the corresponding split with respect to $h_*\xi$. So $h_*$ maps $\Xi$ onto $\Xi$, and since $h_*(k_*\xi_1)=(hk)_*\xi_1$ for $h,k\in\iso$,  we have an action of $\iso$.

The second map, $\xi\mapsto N\xi$, amounts to taking linear combinations of the Killing fields $\xi_1,\xi_2$. The relation \eqref{N-flow} holds because $\xi_1,\xi_2$ commute and are complete, which entails that the respective flows can be constructed via the exponential map. Since $\det N\neq0$, the two components of $N\xi$ are still linearly independent, and since $E$ \eqref{edge1} is invariant under $\xi\mapsto N\xi$, the splitting $M\cong \Rl\times I\times E$ is the same for $\xi$ and $N\xi$. Hence $N\xi\in\Xi$, {\it i.e.}, $\xi\mapsto N\xi$ is a $\GL(2,\Rl)$-action on $\Xi$, and since the push-forward is linear, it is clear that the two actions commute.

To prove the last statement, we consider the submanifold $E$ \eqref{edge1} together with its induced metric. Since the Killing fields $\xi_1,\xi_2$ are tangent to $E$, their flows are isometries of $E$. Since $h_*\xi=N\xi$ and $E$ is two-dimensional, it follows that $N$ acts as an isometry on the tangent space $T_pE$, $p\in E$. But as $E$ is spacelike and two-dimensional, we can assume without loss of generality that the metric of $T_pE$ is the Euclidean metric, and therefore has the two-dimensional Euclidean group ${\rm E}(2)$ as its isometry group. Thus $N\in\GL(2,\Rl)\cap{\rm E}(2)={\rm O}(2)$, {\it i.e.}, $\det N=\pm1$.
\end{proof}

The $\GL(2,\Rl)$-transformation given by the flip matrix $\Pi:=\left(0\;\;1\atop1\;\;0\right)$ will play a special role later on. We therefore reserve the name {\em inverted Killing pair} of $\xi=(\xi_1,\xi_2)\in\Xi$ for
\begin{align}\label{def:xiprime}
 \xi'&:=\Pi\xi=(\xi_2,\xi_1)\,.
\end{align}
Note that since we consider ordered tuples $(\xi_1,\xi_2)$, the Killing pairs $\xi$ and $\xi'$ are not identical. Clearly, the map $\xi\mapsto\xi'$ is an involution on $\Xi$, {\it i.e.}, $(\xi')'=\xi$.
\\
\\
After these preparations, we turn to the construction of wedge regions in $M$, and begin by specifying their {\em edges}.
\begin{definition}\label{edgedef}
An edge is a subset of $M$ which has the form
\begin{align}\label{def:edge}
 E_{\xi,p}:=\{\flow_{\xi,s}(p)\in M\,:\,s\in\Rl^2\}
\end{align}
for some $\xi\in\Xi$, $p\in M$. Any spacelike vector $n_{\xi,p}\in T_pM$ which completes the gradient of the chosen temporal function and the Killing vectors $\xi_1(p),\xi_2(p)$ to a positively oriented basis $(\nabla\mathcal{T}(p),\xi_1(p),\xi_2(p),n_{\xi,p})$ of $T_pM$ is called an oriented normal of $E_{\xi,p}$.
\end{definition}

It is clear from this definition that each edge is a two-dimensional, spacelike, smooth submanifold of $M$. Our definition of admissible spacetimes $M\cong \Rl\times I\times E$ explicitly restricts the topology of $I$, but not of the edge \eqref{edge1}, which can be homeomorphic to a plane, cylinder, or torus.

Note also that the description of the edge $E_{\xi,p}$ in terms of $\xi$ and $p$ is somewhat redundant: Replacing the Killing fields $\xi_1,\xi_2$ by linear combinations $\xiti:=N\xi$, $N\in\GL(2,\Rl)$, or replacing $p$ by $\pti:=\varphi_{\xi,u}(p)$ with some $u\in\Rl^2$, results in the same manifold $E_{\xiti,\pti}=E_{\xi,p}$.
\\\\
Before we define wedges as connected components of causal complements of edges, we have to prove the following key lemma, from which the relevant properties of wedges follow. For its proof, it might be helpful to visualize the geometrical situation as sketched in Figure \ref{fig:wedge}.

\begin{figure}[htbp]
  \centering
    \includegraphics[width=0.5\textwidth]{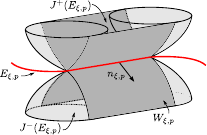}\\
  \caption{\em Three-dimensional sketch of the wedge $W_{\xi,p}$ and its edge $E_{\xi,p}$}
  \label{fig:wedge}
\end{figure}

\begin{lemma}\label{edgeprop}
 The causal complement $E_{\xi,p}'$ of an edge $E_{\xi,p}$ is the disjoint union of two connected components, which are causal complements of each other.
\end{lemma}
\begin{proof}
We first show that any point $q\in E'_{\xi,p}$ is connected to the base point $p$ by a smooth, spacelike curve. Since $M$ is globally hyperbolic, there exist Cauchy surfaces $\Sigma_p,\Sigma_q$ passing through $p$ and $q$, respectively. We pick two compact subsets $K_q\subset \Sigma_q$, containing $q$, and $K_p\subset \Sigma_p$, containing $p$. If $K_p, K_q$ are chosen sufficiently small, their union $K_p\cup K_q$ is an acausal, compact, codimension one submanifold of $M$. It thus fulfils the hypothesis of Thm. 1.1 in \cite{BernalSanchez:2006}, which guarantees that there exists a spacelike Cauchy surface $\Sigma$ containing the said union. In particular, there exists a smooth, spacelike curve $\gamma$ connecting $p=\gamma(0)$ and $q=\gamma(1)$. Picking spacelike vectors $v\in T_p\Sigma$ and $w\in T_q\Sigma$, we have the freedom of choosing $\gamma$ in such a way that $\dot{\gamma}(0)=v$ and $\dot{\gamma}(1)=w$. If $v$ and $w$ are chosen linearly independent from $\xi_1(p),\xi_2(p)$ and $\xi_1(q),\xi_2(q)$, respectively, these vectors are oriented normals of $E_{\xi,p}$ respectively $E_{\xi,q}$, and we can select $\gamma$ such that it intersects the edge $E_{\xi,p}$ only in $p$.

Let us define the region
\begin{align}
\label{wedgechar}
W_{\xi,p}
&:=
\{q\in E_{\xi,p}':\;\exists\, \gamma\in C^1([0,1],M)\text{ with }\gamma(0)=p, \gamma(1)=q, E_{\xi,p}\cap\gamma=\{p\},\nonumber\\
&\qquad
\dot{\gamma}(0) \text{ is an oriented normal of } E_{\xi,p},\,\dot{\gamma}(1) \text{ is an oriented normal of } E_{\xi,q}
\}\,,
\end{align}
and, exchanging $\xi$ with the inverted Killing pair $\xi'$, we correspondingly define the region $W_{\xi',p}$. It is clear from the above argument that $W_{\xi,p}\cup W_{\xi',p}=E_{\xi,p}'$, and that we can prescribe arbitrary normals $n,m$ of $E_{\xi,p}$, $E_{\xi,q}$ as initial respectively final tangent vectors of the curve $\gamma$ connecting $p$ to $q\in W_{\xi,p}$.

The proof of the lemma consists in establishing that $W_{\xi,p}$ and $W_{\xi',p}$ are disjoint, and causal complements of each other. To prove disjointness of $W_{\xi,p}, W_{\xi',p}$, assume there exists a point $q\in W_{\xi,p}\cap W_{\xi',p}$. Then $q$ can be connected with the base point $p$ by two spacelike curves, whose tangent vectors satisfy the conditions in \eqref{wedgechar} with $\xi$ respectively $\xi'$. By joining these two curves, we have identified a continuous loop $\la$ in $E_{\xi,p}'$. As an oriented normal, the tangent vector $\dot{\la}(0)$ at $p$ is linearly independent of $\xi_1(p), \xi_2(p)$, so that $\la$ intersects $E_{\xi,p}$ only in $p$.

Recall that according to Definition \ref{admissible}, $M$ splits as the product $M\cong \Rl\times I\times E_{\xi,p}$, with an open interval $I$ which is smoothly embedded in $M$. Hence we can consider the projection $\pi(\la)$ of the loop $\la$ onto $I$, which is a closed interval $\pi(\la)\subset I$ because the simple connectedness of $I$ rules out the possibility that $\pi(\la)$ forms a loop, and on account of the linear independence of $\{\xi_1(p),\xi_2(p),n_{\xi,p}\}$, the projection cannot be just a single point. Yet, as $\la$ is a loop, there exists $p'\in\la$ such that $\pi(p')=\pi(p)$. We also know that $\pi^{-1}(\{\pi(p)\})=\Rl\times\{\pi(p)\}\times E_{\xi,p}$ is contained in $J^+(E_{\xi,p})\cup E_{\xi,p}\cup J^-(E_{\xi,p})$ and, since $p$ and $p'$ are causally separated, the only possibility left is that they both lie on the same edge. Yet, per construction, we know that the loop intersects the edge only once at $p$ and, thus, $p$ and $p'$ must coincide, which is the sought contradiction.

To verify the claim about causal complements, assume there exist points $q\in W_{\xi,p}$, $q'\in W_{\xi',p}$ and a causal curve $\gamma$ connecting them, $\gamma(0)=q$, $\gamma(1)=q'$. By definition of the causal complement, it is clear that $\gamma$ does not intersect $E_{\xi,p}$. In view of our restriction on the topology of $M$, it follows that $\gamma$ intersects either $J^+(E_{\xi,p})$ or $J^-(E_{\xi,p})$. These two cases are completely analogous, and we consider the latter one, where there exists a point $q''\in \gamma\cap J^-(E_{\xi,p})$. In this situation, we have a causal curve connecting $q\in W_{\xi,p}$ with $q''\in J^-(E_{\xi,p})$, and since $q\notin J^-(q'')\subset J^-(E_{\xi,p})$, it follows that $\gamma$ must be past directed. As the time orientation of $\gamma$ is the same for the whole curve, it follows that also the part of $\gamma$ connecting $q''$ and $q'$ is past directed. Hence $q'\in J^-(q'')\subset J^-(E_{\xi,p})$, which is a contradiction to $q'\in W_{\xi',p}$. Thus $W_{\xi',p}\subset {W_{\xi,p}}'$.

To show that $W_{\xi',p}$ coincides with ${W_{\xi,p}}'$, let $q\in {W_{\xi,p}}'=\overline{W_{\xi,p}}'\subset E_{\xi,p}' = W_{\xi,p}\sqcup W_{\xi',p}$. Yet $q\in W_{\xi,p}$ is not possible since $q\in {W_{\xi,p}}'$ and $W_{\xi,p}$ is open. So $q\in W_{\xi',p}$, {\em i.e.}, we have shown ${W_{\xi,p}}'\subset W_{\xi',p}$, and the claimed identity $W_{\xi',p} = {W_{\xi,p}}'$ follows.
\end{proof}

\begin{wrapfigure}{r}{46mm}
 \centering
{\includegraphics{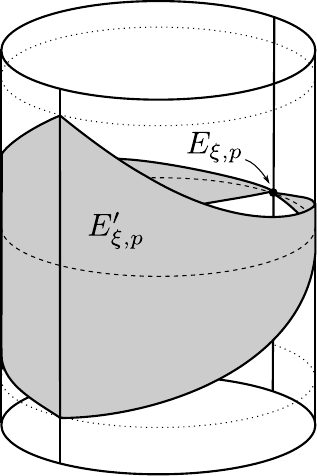}}\\
{\em \small $E_{\xi,p}'$ in a Lorentz cylinder}
\end{wrapfigure}

Lemma \ref{edgeprop} does not hold if the topological requirements on $M$ are dropped. As an example, consider a cylinder universe $\Rl\times S^1\times\Rl^2$, the product of the Lorentz cylinder $\Rl\times S^1$ \cite{O'Neill:1983} and the Euclidean plane $\Rl^2$. The translations in the last factor $\Rl^2$ define spacelike, complete, commuting, linearly independent Killing fields $\xi$.  Yet the causal complement of the edge $E_{\xi,p}=\{0\}\times\{1\}\times\Rl^2$ has only a single connected component, which has empty causal complement. In this situation, wedges lose many of the useful properties which we establish below for admissible spacetimes.
\\\\
In view of Lemma \ref{edgeprop}, wedges in $M$ can be defined as follows.

\begin{definition}{\bf (Wedges)}\\
A wedge is a subset of $M$ which is a connected component of the causal complement of an edge in $M$. Given $\xi\in\Xi$, $p\in M$, we denote by $W_{\xi,p}$ the component of $E_{\xi,p}'$ which intersects the curves $\gamma(t) := \exp_p(t\,n_{\xi,p})$, $t>0$, for any oriented normal $n_{\xi,p}$ of $E_{\xi,p}$. The family of all wedges is denoted
\begin{align}
  \W
:=
\{W_{\xi,p}\,:\, \xi\in\Xi,\,p\in M\}\,.
\end{align}
\end{definition}

As explained in the proof of Lemma \ref{edgeprop}, the condition that the curve $\Rl^+\ni t\mapsto \exp_p(t\,n_{\xi,p})$ intersects a connected component of $E_{\xi,p}'$ is independent of the chosen normal $n_{\xi,p}$, and each such curve intersects precisely one of the two components of $E_{\xi,p}'$.

Some properties of wedges which immediately follow from the construction carried out in the proof of Lemma \ref{edgeprop} are listed in the following proposition.

\begin{proposition}{\bf (Properties of wedges)}\label{prop:wedge-properties}
\\
 Let $W=W_{\xi,p}$ be a wedge. Then
\begin{enumerate}
 \item $W$ is causally complete, {\it i.e.}, $W''=W$, and hence globally hyperbolic.
\item The causal complement of a wedge is given by inverting its Killing pair,
\begin{align}\label{cc-inverted}
 (W_{\xi,p})' = W_{\xi',p}\,.
\end{align}
 \item A wedge is invariant under the Killing flow generating its edge,
\begin{align}
\flow_{\xi,s}(W_{\xi,p})=W_{\xi,p}\,,\qquad s\in\Rl^2\,.
\end{align}
\end{enumerate}
\end{proposition}
\begin{proof}
{\em a)} By Lemma \ref{edgeprop}, $W$ is the causal complement of another wedge $V$, and therefore causally complete: $W''=V'''=V'=W$. Since $M$ is globally hyperbolic, this implies that $W$ is globally hyperbolic, too \cite[Prop. 12.5]{Keyl:1996}.

{\em b)} This statement has already been checked in the proof of Lemma \ref{edgeprop}.

{\em c)} By definition of the edge $E_{\xi,p}$ \eqref{def:edge}, we have $\flow_{\xi,s}(E_{\xi,p})= E_{\xi,p}$ for any $s\in\Rl^2$, and since the $\flow_{\xi,s}$ are isometries, it follows that $\flow_{\xi,s}(E_{\xi,p}')= E_{\xi,p}'$. Continuity of the flow implies that also the two connected components of this set are invariant.
\end{proof}

\begin{corollary}{\bf (Properties of the family of wedge regions)}\\
 The family $\W$ of wedge regions is invariant under the isometry group $\iso$ and under taking causal complements. For $h\in \iso$, it holds
\begin{align}\label{h-W}
 h(W_{\xi,p})=W_{h_*\xi,h(p)}\,.
\end{align}
\end{corollary}
\begin{proof}
Since isometries preserve the causal structure of a spacetime, we only need to look at the action of  isometries on edges. We find
\begin{equation}
hE_{\xi,p}
=\{h\circ\flow_{\xi,s}\circ h^{-1}(h(p)):s\in\Rl^2\}
=\{\flow_{h_*\xi,s}(h(p)):s\in\Rl^2\}
=E_{h_*\xi,h(p)}
\end{equation}
by using the well-known fact that conjugation of flows by isometries amounts to the push-forward by the isometry of the associated vector field. Since $h_*\xi\in\Xi$ for any $\xi\in\Xi$, $h\in \iso$ (Lemma \ref{lemma:group-xi}), the family $\W$ is invariant under the action of the isometry group. Closedness of $\W$ under causal complementation is clear from Prop. \ref{prop:wedge-properties} b).
\end{proof}

In contrast to the situation in flat spacetime, the isometry group $\iso$ does {\em not} act transitively on $\W(M,g)$ for generic admissible $M$, and there is no isometry mapping a given wedge onto its causal complement. This can be seen explicitly in the examples discussed in Section \ref{sec:examples}. To keep track of this structure of $\W(M,g)$, we decompose $\Xi(M,g)$ into orbits under the $\iso$- and $\GL(2,\Rl)$-actions.

\begin{definition}\label{def:equivalence}
Two Killing pairs $\xi,\xiti\in\Xi$ are equivalent, written $\xi\sim\xiti$, if there exist $h\in\iso$ and $N\in\GL(2,\Rl)$ such that $\xiti = N h_* \xi$.
\end{definition}
As $\xi\mapsto N\xi$ and $\xi\mapsto h_*\xi$ are commuting group actions, $\sim$ is an equivalence relation. According to Lemma \ref{lemma:group-xi} and Prop. \ref{prop:wedge-properties} b), c), acting with $N\in\GL(2,\Rl)$ on $\xi$ either leaves $W_{N\xi,p}=W_{\xi,p}$ invariant (if $\det N>0$) or exchanges this wedge with its causal complement, $W_{N\xi,p}=W_{\xi,p}'$ (if $\det N<0$). Therefore the ``coherent''\footnote{See \cite{BuchholzSummers:2007} for a related notion on Minkowski space.} subfamilies arising in the decomposition of the family of all wedges along the equivalence classes $[\xi]\in\Xi\slash\!\!\sim$,
\begin{align}\label{W-decomposition}
 \W = \bigsqcup_{[\xi]}\W_{[\xi]}\,,
\qquad
\W_{[\xi]}
:=
\{W_{\xiti,p}\,:\,\xiti\sim\xi,\,p\in M\}\,,
\end{align}
take the form
\begin{align}\label{def:Wxi}
\W_{[\xi]}
=
\{W_{h_*\xi,p},\,W_{h_*\xi,p}'\,:\,h\in\iso,\,p\in M\}
\,.
\end{align}
In particular, each subfamily $\W_{[\xi]}$ is invariant under the action of the isometry group and causal complementation.

In our later applications to quantum field theory, it will be important to have control over causal configurations $W_1\subset W_2'$ and inclusions $W_1\subset W_2$ of wedges $W_1,W_2\in\W$. Since $\W$ is closed under taking causal complements, it is sufficient to consider inclusions. Note that the following proposition states in particular that inclusions can only occur between wedges from the same coherent subfamily $\W_{[\xi]}$.

\begin{proposition}{\bf (Inclusions of wedges).}\label{prop:inclusions}\\
Let $\xi,\xiti\in\Xi$, $p,\pti\in M$. The wedges $W_{\xi,p}$ and $W_{\xiti,\pti}$ form an inclusion, $W_{\xi,p}\subset W_{\xiti,\pti}$, if and only if $p\in \overline{W_{\xiti,\pti}}$ and there exists $N\in\GL(2,\Rl)$ with $\det N>0$, such that $\xiti=N\xi$.
\end{proposition}
\begin{proof}
($\Leftarrow$) Let us assume that $\xiti=N\xi$ holds for some $N\in\GL(2,\Rl)$ with $\det N>0$,  and $p\in \overline{W_{\xiti,\pti}}$. In this case, the Killing fields in $\xiti$ are linear combinations of those in $\xi$, and consequently, the edges $E_{\xi,p}$ and $E_{\xiti,\pti}$ intersect if and only if they coincide, {\it i.e.} if $\pti\in E_{\xi,p}$. If the edges coincide, we clearly have $W_{\xiti,\pti}=W_{\xi,p}$. If they do not coincide, it follows from $p\in \overline{W_{\xiti,\pti}}$ that $E_{\xi,p}$ and $E_{\xiti,\pti}$ are either spacelike separated or they can be connected by a null geodesic.

Consider now the case that $E_{\xi,p}$ and $E_{\xiti,\pti}$ are spacelike separated, {\it i.e.} $p\in W_{\xiti,\pti}$. Pick a point $q\in W_{\xi,p}$ and recall that $W_{\xi,p}$ can be characterized by equation (\ref{wedgechar}). Since  $p\in W_{\xiti,\pti}$ and $q\in W_{\xi,p}$, there exist curves $\gamma_p$ and $\gamma_q$, which connect the pairs of points $(\pti,p)$ and $(p,q)$, respectively, and comply with the conditions in (\ref{wedgechar}). By joining $\gamma_{p}$ and $\gamma_q$ we obtain a curve which connects $\pti$ and $q$. The tangent vectors $\dot{\gamma}_p(1)$ and $\dot{\gamma}_q(0)$ are oriented normals of $E_{\xi, p}$ and we choose $\gamma_p$ and $\gamma_q$ in such a way that these tangent vectors coincide. Due to the properties of $\gamma_p$ and $\gamma_q$, the joint curve also complies with the conditions in (\ref{wedgechar}), from which we conclude $q\in W_{\xiti,\pti}$, and thus $W_{\xi,p}\subset W_{\xiti,\pti}$.

Consider now the case that $E_{\xiti,\pti}$ and $E_{\xi,p}$ are connected by null geodesics, {\it i.e.} $p\in \partial\overline{W_{\xiti,\pti}}$. Let $r$ be the point in $E_{\xi,p}$ which can be connected by a null geodesic with $\pti$ and pick a point $q\in W_{\xi,p}$. The intersection $J^-(r)\cap \partial\overline{W_{\xi,p}}$ yields another null curve, say $\mu$, and the intersection $\mu\cap J^-(q)=:p'$ is non-empty since $r$ and $q$ are spacelike separated and $q\in W_{\xi,p}$. The null curve $\mu$ is chosen future directed and parametrized in such a way that $\mu(0)=p'$ and $\mu(1)=r$. By taking $\varepsilon\in(0,1)$ we find $q\in W_{\xi,\mu(\varepsilon)}$ and $\mu(\varepsilon)\in W_{\xiti,\pti}$ which entails $q\in W_{\xiti,\pti}$.

($\Rightarrow$) Let us assume that we have an inclusion of wedges $W_{\xi,p} \subset W_{\xiti,\pti}$. Then clearly $p\in\overline{W_{\xiti,\pti}}$. Since $M$ is four-dimensional and $\xi_1,\xi_2,\xiti_1,\xiti_2$ are all spacelike, they cannot be linearly independent. Let us first assume that three of them are linearly independent, and without loss of generality, let $\xi=(\xi_1,\xi_2)$ and $\xiti=(\xi_2, \xi_3)$ with three linearly independent spacelike Killing fields $\xi_1,\xi_2,\xi_3$. Picking points $q\in E_{\xi,p}$, $\tilde{q}\in E_{\xiti,\pti}$ these can be written as $q=(t,x_1,x_2,x_3)$ and $\tilde{q}=(\tilde{t},\tilde{x}_1,\tilde{x}_2,\tilde{x}_3)$ in the global coordinate system of flow parameters constructed from $\xi_1,\xi_2,\xi_3$ and the gradient of the temporal function.

For suitable flow parameters $s_1,s_2,s_3$, we have $\flow_{\xi_1,s_1}(q)=(t,\tilde{x}_1,x_2,x_3)=:q'\in E_{\xi,p}$ and $\flow_{(\xi_2,\xi_3),(s_2,s_3)}(\tilde{q})=(\tilde{t},\tilde{x}_1,x_2,x_3)=:\tilde{q}'\in E_{\xiti,\pti}$. Clearly, the points $q'$ and $\tilde{q}'$ are  connected by a timelike curve, {\it e.g.} the curve whose tangent vector field is given by the gradient of the temporal function. But a timelike curve connecting the edges of $W_{\xi,p}, W_{\xiti,\pti}$ is a contradiction to these wedges forming an inclusion. So no three of the vector fields $\xi_1,\xi_2,\xiti_1,\xiti_2$ can be linearly independent.

Hence $\xiti= N\xi$ with an invertible matrix $N$. It remains to establish the correct sign of $\det N$, and to this end, we assume $\det N<0$. Then we have $(W_{\xi,p})'=W_{\xi',p}\subset W_{\xiti,\pti}$, by  (Prop. \ref{prop:wedge-properties} b)) and the ($\Leftarrow$) statement in this proof, since $\xiti$ and $\xi'$ are related by a positive determinant transformation and $p\in \overline{W_{\xiti,\pti}}$. This yields that both, $W_{\xi,p}$ and its causal complement, must be contained in $W_{\xiti,\pti}$, a contradiction. Hence $\det N>0$, and the proof is finished.
\end{proof}

Having derived the structural properties of the set $\W$ of wedges needed later, we now compare our wedge regions to the Minkowski wedges and to other definitions proposed in the literature.

The flat Minkowski spacetime $(\Rl^4,\eta)$ clearly belongs to the class of admissible spacetimes, with translations along spacelike directions and rotations in the standard time zero Cauchy surface as its complete spacelike Killing fields. However, as Killing pairs consist of non-vanishing vector fields, and each rotation leaves its rotation axis invariant, the set $\Xi(\Rl^4,\eta)$ consists precisely of all pairs $(\xi_1,\xi_2)$ such that the flows $\flow_{\xi_1}$, $\flow_{\xi_2}$ are translations along two linearly independent spacelike directions. Hence the set of all edges in Minkowski space coincides with the set of all two-dimensional spacelike planes. Consequently, each wedge $W\in\W(\Rl^4,\eta)$ is bounded by two non-parallel characteristic three-dimensional planes. This is precisely the family of wedges usually considered in Minkowski space\footnote{Note that we would get a ``too large'' family of wedges in Minkowski space if we would drop the requirement that the vector fields generating edges are Killing. However, the assumption that edges are generated by {\em commuting} Killing fields is motivated by the application to deformations of quantum field theories, and one could generalize our framework to spacetimes with edges generated by complete, linearly independent smooth Killing fields.} (see, for example, \cite{ThomasWichmann:1997}).

Besides the features we established above in the general admissible setting, the family of Minkowski wedges has the following well-known properties:

\begin{enumerate}
 \item Each wedge $W\in\W(\Rl^4,\eta)$ is the causal completion of the world line of a uniformly accelerated observer.
 \item Each wedge $W\in\W(\Rl^4,\eta)$ is the union of a family of double cones whose tips lie on two fixed lightrays.
\item The isometry group (the Poincar\'e group) acts transitively on $\W(\Rl^4,\eta)$.
 \item $\W(\Rl^4,\eta)$ is {\em causally separating}\label{causal-sep} in the sense that given
any two spacelike separated double cones $\OO_1,\OO_2\subset\Rl^4$, then there exists a wedge $W$ such that $\OO_1\subset
W\subset\OO_2'$ \cite{ThomasWichmann:1997}. $\W(\Rl^4,\eta)$ is a subbase for the topology of $\Rl^4$.
\end{enumerate}

All these properties a)--d) do {\em not} hold for the class $\W(M,g)$ of wedges on a {\em general} admissible spacetime, but some hold for certain subclasses, as can be seen from the explicit examples in the subsequent section.

There exist a number of different constructions for wedges in curved spacetimes in the literature, mostly for special manifolds. On de Sitter respectively anti de Sitter space Borchers and Buchholz \cite{BorchersBuchholz:1999} respectively Buchholz and Summers \cite{BuchholzSummers:2004-2} construct wedges by taking property a) as their defining feature, see also the generalization by Strich \cite{Strich:2008}. In the de Sitter case, this definition is equivalent to our definition of a wedge as a connected component of the causal complement of an edge \cite{BuchholzMundSummers:2001}. But as two-dimensional spheres, the de Sitter edges do not admit two linearly independent commuting Killing fields. Apart from this difference due to our restriction to commuting, linearly independent, Killing fields, the de Sitter wedges can be constructed  in the same way as presented here. Thanks to the maximal symmetry of the de Sitter and anti de Sitter spacetimes, the respective isometry groups act transitively on the corresponding wedge families (c), and causally separate in the sense of d).

A definition related to the previous examples has been given by Lauridsen-Ribeiro for wedges in asymptotically anti de Sitter spacetimes (see Def. 1.5 in \cite{LauridsenRibeiro:2007}). Note that these spacetimes are not admissible in our sense since anti de Sitter space is not globally hyperbolic.

Property b) has recently been taken by Borchers \cite{Borchers:2009} as a definition of wedges in a quite general class of curved spacetimes which is closely related to the structure of double cones. In that setting, wedges do not exhibit in general all of the features we derived in our framework, and can for example have compact closure.

Wedges in a class of Friedmann-Robertson-Walker spacetimes with spherical spatial sections have been constructed with the help of conformal embeddings into de Sitter space \cite{BuchholzMundSummers:2001}. This construction also yields wedges defined as connected components of causal complements of edges. Here a) does not, but c) and d) do hold, see also our discussion of Friedmann-Robertson-Walker spacetimes with flat spatial sections in the next section.

The idea of constructing wedges as connected components of causal complements of specific two-dimensional submanifolds has also been used in the context of globally hyperbolic spacetimes with a bifurcate Killing horizon \cite{GuidoLongoRobertsVerch:2001}, building on earlier work in \cite{Kay:1985}. Here the edge is given as the fixed point manifold of the Killing flow associated with the horizon.

 \subsection{Concrete examples}\label{sec:examples}

In the previous section we provided a complete but abstract characterization of the geometric structures of the class of spacetimes we are interested in. This analysis is now complemented by presenting a number of explicit examples of admissible spacetimes.

The easiest way to construct an admissble spacetime is to take the warped product \cite[Chap. 7]{O'Neill:1983} of an edge with another manifold. Let $(E,g_E)$ be a two-dimensional Riemannian manifold endowed with two complete, commuting, linearly independent, smooth Killing fields, and let $(X,g_X)$ be a two-dimensional, globally hyperbolic spacetime diffeomorphic to $\Rl\times I$, with $I$ an open interval or the full real line. Then, given a positive smooth function $f$ on $X$, consider the {\em warped product}
$M:=X\times_f E$, {\em i.e.}, the product manifold $X\times E$ endowed with the metric tensor field
\begin{align*}
g:=\pi_X^*(g_X)+(f\circ\pi_X) \cdot \pi_E^*(g_E),
\end{align*}
where $\pi_X:M\to X$ and $\pi_E:M\to E$ are the projections on $X$ and $E$. It readily follows that $(M,g)$ is admissible in the sense of Definition \ref{admissible}.
\\\\
The following proposition describes an explicit class of admissible spacetimes in terms of their metrics.
\begin{proposition}\label{prop:metric}
Let $(M,g)$ be a spacetime diffeomorphic to $\Rl\times I\times \Rl^2$, where $I\subseteq\Rl$ is open and simply connected, endowed with a global coordinate system $(t,x,y,z)$ according to which the metric reads
\begin{equation}\label{metric}
 ds^2=e^{2f_0}dt^2-e^{2f_1}dx^2-e^{2f_2}dy^2-e^{2f_3}(dz-q\,dy)^2.
\end{equation}
Here $t$ runs over the whole $\Rl$, $f_i,q\in C^\infty(M)$ for $i=0,...,3$ and $f_i,q$ do not depend on $y$ and $z$. Then $(M,g)$ is an admissible spacetime in the sense of Definition \ref{admissible}.
\end{proposition}
\begin{proof}
Per direct inspection of \eqref{metric}, $M$ is isometric to $\Rl\times\Sigma$ with $\Sigma\cong I\times\Rl^2$ with $ds^2=\beta\, dt^2-h_{ij}dx^i dx^j$, where $\beta$ is smooth and positive, and $h$ is a metric which depends smoothly on $t$. Furthermore, on the hypersurfaces at constant $t$, $\det h=e^{2(f_1+f_2+f_3)}>0$ and $h$ is block-diagonal. If we consider the sub-matrix with $i,j=y,z$, this has a positive determinant and a positive trace. Hence we can conclude that the induced metric on $\Sigma$ is Riemannian, or, in other words, $\Sigma$ is a spacelike, smooth, three-dimensional Riemannian hypersurface. Therefore we can apply Theorem 1.1 in \cite{BernalSanchez:2005} to conclude that $M$ is globally hyperbolic.

Since the metric coefficients are independent from $y$ and $z$, the vector fields $\xi_1=\partial_y$ and $\xi_2=\partial_z$ are smooth Killing fields which commute and, as they lie tangent to the Riemannian hypersurfaces at constant time, they are also spacelike. Furthermore, since per definition of spacetime, $M$ and thus also $\Sigma$ is connected, we can invoke the Hopf-Rinow-Theorem \cite[$\S$ 5, Thm. 21]{O'Neill:1983} to conclude that $\Sigma$ is complete and, thus, all its Killing fields are complete. As $I$ is simply connected by assumption, it follows that $(M,g)$ is admissible.
\end{proof}

Under an additional assumption, also a partial converse of Proposition \ref{prop:metric} is true. Namely, let $(M,g)$ be a globally hyperbolic spacetime with two complete, spacelike, commuting, smooth Killing fields, and pick a local coordinate system $(t,x,y,z)$, where $y$ and $z$ are the flow parameters of the Killing fields. Then, if the reflection map $r:M\to M$, $r(t,x,y,z)=(t,x,-y,-z)$, is an isometry, the metric is locally of the form \eqref{metric}, as was proven in \cite{Chandrasekhar:1983, ChandrasekharFerrari:1984}. The reflection $r$ is used to guarantee the vanishing of the unwanted off-diagonal metric coefficients, namely those associated to ``$dx\,dy$" and ``$dx\,dz$". Notice that the cited papers allow only to establish a result on the local structure of $M$ and no a priori condition is imposed on the topology of $I$, in distinction to Proposition \ref{prop:metric}.
\\\\
Some of the metrics \eqref{metric} are used in cosmology. For the description of a spatially homogeneous but in general anisotropic universe $M\cong J\times\Rl^3$ where $J\subseteq\Rl$ (see \S 5 in \cite{Wald:1984} and \cite{FullingParkerHu:1974}), one puts $f_0=q=0$ in \eqref{metric} and takes $f_1,f_2,f_3$ to depend only on $t$. This yields the metric of {\em Kasner spacetimes} respectively {\em Bianchi I models}\footnote{
The Bianchi models I--IX \cite{Ellis:2006} arise from the classification of three-dimensional real Lie algebras, thought of as Lie subalgebras of the Lie algebra of Killing fields. Only the cases  Bianchi I--VII, in which the three-dimensional Lie algebra contains $\Rl^2$ as a subalgebra, are of direct interest here, since only in these cases Killing pairs exist.}
 \begin{equation}\label{Kasner}
 ds^2=dt^2-e^{2f_1}dx^2-e^{2f_2}dy^2-e^{2f_3}dz^2\,.
 \end{equation}
Clearly here the isometry group contains three smooth Killing fields, locally given by $\partial_x, \partial_y, \partial_z$, which are everywhere linearly independent, complete and commuting. In particular,   $(\partial_x,\partial_y)$, $(\partial_x,\partial_z)$ and $(\partial_y,\partial_z)$ are Killing pairs.
\\\\
A case of great physical relevance arises when specializing the metric further by taking all the functions $f_i$ in \eqref{Kasner} to coincide. In this case, the metric assumes the so-called  {\em Friedmann-Robertson-Walker} form
 \begin{equation}\label{FRW}
 ds^2
= dt^2-a(t)^2\,[dx^2+dy^2+dz^2]
=a(t(\tau))^2\,\left[d\tau^2-dx^2-dy^2-dz^2\right]
\,.
 \end{equation}
Here the {\em scale factor} $a(t) := e^{f_1(t)}$ is defined on some interval $J\subseteq\Rl$, and in the second equality, we have introduced the {\em conformal time} $\tau$, which is implicitely defined by $d\tau = a^{-1}(t)dt$. Notice that, as in the Bianchi I model, the manifold is $M\cong J\times\Rl^3$, {\em i.e.}, the variable $t$ does not need to range over the whole real axis. (This does not affect the property of global hyperbolicity.)

By inspection of \eqref{FRW}, it is clear that the isometry group of this spacetime contains the three-dimensional Euclidean group ${\rm E}(3) = {\rm O}(3)\rtimes\Rl^3$. Disregarding the Minkowski case, where $J=\Rl$ and $a$ is constant, the isometry group in fact coincides with ${\rm E}(3)$. Edges in such a Friedmann-Robertson-Walker universe are of the form $\{\tau\}\times S$, where $S$ is a two-dimensional plane in $\Rl^3$ and $t(\tau)\in J$. Here $\W$ consists of a single coherent family, and the $\iso$-orbits in $\W$ are labelled by the time parameter $\tau$ for the corresponding edges. Also note that the family of Friedmann-Robertson-Walker wedges is causally separating in the sense discussed on page \pageref{causal-sep}.

The second form of the metric in \eqref{FRW} is manifestly a conformal rescaling of the flat Minkowski metric. Interpreting the coordinates $(\tau,x,y,z)$ as coordinates of a point in $\Rl^4$ therefore gives rise to a conformal embedding $\iota:M\to\Rl^4$.

\begin{wrapfigure}{r}{60mm}
 \centering
{\includegraphics{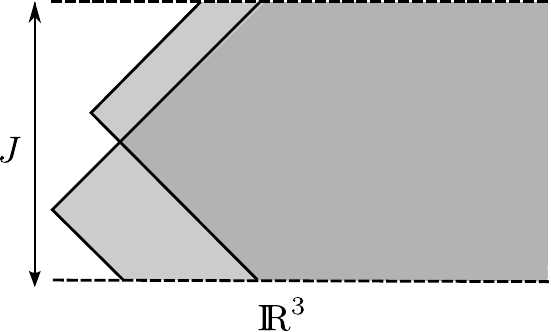}}\\
{\em Two wedges in FRW spacetime}
\end{wrapfigure}
\noindent
In this situation, it is interesting to note that the set of all images $\iota(E)$ of edges $E$ in the Friedmann-Robertson-Walker spacetime coincides with the set of all Minkowski space edges which lie completely in $\iota(M)=J\times\Rl^3$, provided that $J$ does not coincide with $\Rl$. These are just the edges parallel to the standard Cauchy surfaces of constant $\tau$ in $\Rl^4$.
So Friedmann-Robertson-Walker edges can also be characterized in terms of Minkowski space edges and the conformal embedding $\iota$, analogous to the construction of wedges in Friedmann-Robertson-Walker spacetimes with spherical spatial sections in \cite{BuchholzMundSummers:2001}.


\section{Quantum field theories on admissible spacetimes}\label{sec:deformation}

Having discussed the relevant geometric structures, we now fix an admissible spacetime $(M,g)$ and discuss warped convolution deformations of quantum field theories on it. For models on flat Minkowski space, it is known that this deformation procedure weakens point-like localization to localization in wedges \cite{BuchholzLechnerSummers:2010}, and we will show here that the same holds true for admissible curved spacetimes. For a convenient description of this weakened form of localization, and for a straightforward application of the warped convolution technique, we will work in the framework of local quantum physics \cite{Haag:1996}.

In this setting, a model theory is defined by a net of field algebras, and here we consider algebras $\fF(W)$ of quantum fields supported in wedges $W\in\W(M,g)$. The main idea underlying the deformation is to apply the formalism developed in \cite{BuchholzSummers:2008,BuchholzLechnerSummers:2010}, but with the global translation symmetries of Minkowski space replaced by the Killing flow $\flow_\xi$ corresponding to the wedge $W=W_{\xi,p}$ under consideration. In the case of Minkowski spacetime, these deformations reduce to the familiar structure of a noncommutative Minkowski space with commuting time.

The details of the model under consideration will not be important in Section \ref{sec:generaldeformations}, since our construction relies only on a few structural properties satisfied in any well-behaved quantum field theory. In Section \ref{sec:dirac}, the deformed Dirac quantum field  is presented as a particular example.

\subsection{Deformations of nets with Killing symmetries}\label{sec:generaldeformations}

Proceeding to the standard mathematical formalism \cite{Haag:1996,Araki:1999}, we consider a $C^*$-algebra $\fF$, whose elements are interpreted as (bounded functions of) quantum fields on the spacetime $M$. The field algebra $\fF$ has a local structure, and in the present context, we focus on localization in wedges $W\in\W$, since this form of localization turns out to be stable under the deformation. Therefore, corresponding to each wedge $W\in\W$, we consider the $C^*$-subalgebra $\fF(W)\subset\fF$ of fields supported in $W$. Furthermore, we assume a strongly continuous action $\alpha$ of the isometry group $\iso$ of $(M,g)$ on $\fF$, and a Bose/Fermi automorphism $\gamma$ whose square is the identity automorphism, and which commutes with $\alpha$. This automorphism will be used to separate the Bose/Fermi parts of fields $F\in\fF$; in the model theory of a free Dirac field discussed later, it can be chosen as a rotation by $2\pi$ in the Dirac bundle.

To allow for a straightforward application of the results of \cite{BuchholzLechnerSummers:2010}, we will also assume in the following that the field algebra is concretely realized on a separable Hilbert space $\Hil$, which carries a unitary representation $U$ of $\iso$ implementing the action $\alpha$, {\em i.e.},
\begin{align*}
U(h)FU(h)^{-1} = \alpha_h(F)
\,,\qquad
h\in\iso,\,F\in\fF\,.
\end{align*}
We emphasize that despite working on a Hilbert space, we do not select a state, since we do not make any assumptions regarding $U$-invariant vectors in $\Hil$ or the spectrum of subgroups of the representation $U$.\footnote{Note that every $C^*$-dynamical system $(\A,G,\alpha)$, where $\A\subset \mathcal{B(H)}$ is a concrete $C^*$-algebra on a separable Hilbert space $\mathcal{H}$ and $\alpha:G\to \Aut(\A)$ is a strongly continuous representation of the locally compact group $G$, has a covariant representation \cite[Prop. 7.4.7, Lemma 7.4.9]{Pedersen:1979}, build out of the left-regular representation on the Hilbert space $L^2(G)\otimes \mathcal{H}$.} The subsequent analysis will be carried out in a $C^*$-setting, without using the weak closures of the field algebras $\fF(W)$ in $\B(\Hil)$.

For convenience, we also require the Bose/Fermi automorphism $\gamma$ to be unitarily implemented on $\Hil$, {\em i.e.}, there exists a unitary $V=V^*=V^{-1}\in\B(\Hil)$ such that $\gamma(F)=VFV$. We will also use the associated unitary twist operator
\begin{align}\label{def:Z}
 Z := \frac{1}{\sqrt{2}}(1-iV)
\,.
\end{align}
Clearly, the unitarily implemented $\alpha$ and $\gamma$ can be continued to all of $\B(\Hil)$. By a slight abuse of notation, these extensions will be denoted by the same symbols.

In terms of the data $\{\fF(W)\}_{W\in\W}, \alpha, \gamma$, the structural properties of a quantum field theory on $M$ can be summarized as follows \cite{Haag:1996,Araki:1999}.

\begin{enumerate}
 \item {\em Isotony:} $\fF(W)\subset\fF(\tilde{W})$ whenever $W\subset\tilde{W}$.
 \item {\em Covariance} under $\iso$:
\begin{align}
 \alpha_h(\fF(W)) = \fF(h W)\,,\qquad h\in \iso,\;W\in\W\,.\label{covariance}
\end{align}
\item {\em Twisted Locality:} With the unitary $Z$ \eqref{def:Z},  there holds
\begin{align}\label{twisted-locality}
[ZFZ^*,\,G]=0\qquad  \text{for }\;\,F\in\fF(W), G\in\fF(W'),\;W\in\W\,.
\end{align}
\end{enumerate}

The twisted locality condition \eqref{twisted-locality} is equivalent to normal commutation relations between the Bose/Fermi parts $F_\pm:=\frac{1}{2}(F\pm \gamma(F))$ of fields in spacelike separated wedges, $[F_+,G_\pm]=[F_\pm,G_+]=\{F_-,G_-\}=0$ for $F\in\fF(W), G\in\fF(W')$ \cite{DoplicherHaagRoberts:1969}.

The covariance requirement \eqref{covariance} entails that for any Killing pair $\xi\in\Xi$, the algebra $\fF$ carries a corresponding $\Rl^2$-action $\tau_\xi$, defined by
\begin{align*}
  \tau_{\xi,s}:=\alpha_{\flow_{\xi,s}}=\text{ad}\,U_\xi(s)\,,\qquad s\in\Rl^2\,,
\end{align*}
where $U_\xi(s)$ is shorthand for $U(\flow_{\xi,s})$. Since a wedge of the form $W_{\xi,p}$ with some $p\in M$ is invariant under the flows $\varphi_{N\xi,s}$ for any $N\in\GL(2,\Rl)$ (see Prop. \ref{prop:wedge-properties} c) and Lemma \ref{lemma:group-xi}), we have in view of isotony
\begin{align*}
 \tau_{N\xi,s}(\fF(W_{\xi,p}))
=
\fF(W_{\xi,p})
\,,\qquad
N\in\GL(2,\Rl),\; s\in\Rl^2\,.
\end{align*}

In this setting, all structural elements necessary for the application of warped convolution deformations \cite{BuchholzLechnerSummers:2010} are present, and we will use this technique to define a deformed net $W\longmapsto\fF(W)_\la$ of $C^*$-algebras on $M$, depending on a deformation parameter $\la\in\Rl$. For $\la=0$, we will recover the original theory, $\fF(W)_0=\fF(W)$, and for each $\la\in\Rl$, the three basic properties a)--c) listed above will remain valid. To achieve this, the elements of $\fF(W)$ will be deformed with the help of the Killing flow leaving $W$ invariant. We begin by recalling some definitions and results from \cite{BuchholzLechnerSummers:2010}, adapted to the situation at hand.
\\\\
Similar to the Weyl product appearing in the quantization of classical systems, the warped convolution deformation is defined in terms of oscillatory integrals of $\fF$-valued functions, and we have to introduce the appropriate smooth elements first. The action $\alpha$ is a strongly continuous representation of the Lie group $\iso$, which acts automorphically and thus isometrically on the $C^*$-algebra $\fF$. In view of these properties, the smooth elements $\fF^\infty:=\{F\in\fF\,:\,\iso\ni h\mapsto\alpha_h(F)\;\text{is } \|\cdot\|_\fF\text{-smooth}\}$ form a norm-dense ${}^*$-subalgebra $\fF^\infty\subset\fF$ (see, for example, \cite{Taylor:1986}). However, the subalgebras $\fF(W_{\xi,p})\subset\fF$ are in general only invariant under the $\Rl^2$-action $\tau_\xi$, and we therefore also introduce a weakened form of smoothness. An operator $F\in\fF$ will be called $\xi${\em -smooth} if
\begin{align}
 \Rl^2\ni s\mapsto \tau_{\xi,s}(F)\in\fF
\end{align}
is smooth in the norm topology of $\fF$. On the Hilbert space level, we have a dense domain $\Hil^\infty:=\{\Psi\in\Hil\,:\,\iso\ni h\mapsto U(h)\Psi\;\text{is } \|\cdot\|_\Hil\text{-smooth}\}$ of smooth vectors in $\Hil$.

As further ingredients for the definition of the oscillatory integrals, we pick a smooth, compactly supported ``cutoff'' function $\chi\in C_0^\infty(\Rl^2\times\Rl^2)$ with $\chi(0,0)=1$, and the standard antisymmetric $(2\times2)$-matrix
\begin{align}\label{def:Q}
 Q:=
\left(
\begin{array}{cc}
 0&1\\-1&0
\end{array}
\right)
\,.
\end{align}
With these data, we associate to a $\xi$-smooth $F\in\fF$ the deformed operator ({\em warped convolution}) \cite{BuchholzLechnerSummers:2010}
\begin{align}\label{def:Ala}
F_{\xi,\la}
:=
\frac{1}{4\pi^2}
\lim_{\eps\to0}
\int ds\,ds'\,e^{-iss'}\chi(\eps s,\eps s')\,U_\xi(\la Qs)FU_\xi(s'-\la Qs)
\,,
\end{align}
where $\la$ is a real parameter, and $ss'$ denotes the standard Euclidean inner product of $s,s'\in\Rl^2$. The above limit exists in the strong operator topology of $\B(\Hil)$ on the dense subspace $\Hil^\infty$, and is independent of the chosen cutoff function $\chi$ within the specified class. The thus (densely) defined operator $F_{\xi,\la}$ can be shown to extend to a bounded $\xi$-smooth operator on all of $\Hil$, which we denote by the same symbol \cite{BuchholzLechnerSummers:2010}. As can be seen from the above formula, setting $\la=0$ yields the undeformed operator $F_{\xi,0}=F$, for any $\xi\in\Xi$.

The deformation $F\to F_{\xi,\la}$ is closely related to Rieffel's deformation of $C^*$-algebras \cite{Rieffel:1992}, where one introduces the deformed product
\begin{align}\label{rieffel-product}
 F\times_{\xi,\la} G :=
\frac{1}{4\pi^2}
\lim_{\eps\to0}
\int ds\,ds'\,e^{-iss'}\chi(\eps s,\eps s')\,\tau_{\xi,\la Qs}(F)\tau_{\xi,s'}(G)
\,.
\end{align}
This limit exists in the norm topology of $\fF$ for any $\xi$-smooth $F,G\in\fF$, and $F\times_{\xi,\la} G$ is $\xi$-smooth as well.

As is well known, this procedure applies in particular to the deformation of classical theories in terms of star products. As field algebra, one would then take a suitable commutative ${}^*$-algebra of functions on $M$, endowed with the usual pointwise operations. The isometry group acts on this algebra automorphically by pullback, and in particular, the flow $\flow_{\xi}$ of any Killing pair $\xi\in\Xi$ induces automorphisms. The Rieffel product therefore defines a star product on the subalgebra of smooth elements $f,g$ for this action,
\begin{align}
(f\star_{\xi,\la} g)(p)
=
\frac{1}{4\pi^2}
\lim_{\eps\to0}
\int d^2s\,d^2s' e^{-iss'}\,\chi(\eps s,\eps s')\, f(\flow_{\xi,\la Q s}(p))g(\flow_{\xi,s'}(p))
\,.
\end{align}
The function algebra endowed with this star product can be interpreted as a noncommutative version of the manifold $M$, similar to the flat case \cite{GayralGraciaBondiaIochumSchuckerVarilly:2004}. Note that since we are using a two-dimensional spacelike flow on a four-dimensional spacetime, the deformation corresponds to a noncommutative Minkowski space with ``commuting time'' in the flat case.
\\\\
The properties of the deformation map $F\to F_{\xi,\la}$ which will be relevant here are the following.
\begin{lemma}{\bf \cite{BuchholzLechnerSummers:2010}:}\label{thm:deformationproperties}
\\
Let $\xi\in\Xi$, $\la\in\Rl$, and consider $\xi$-smooth operators
$F,G\in\fF$. Then
 \begin{enumerate}
  \item ${F_{\xi,\la}}^*={F^*}_{\xi,\la}$.
  \item $F_{\xi,\la} G_{\xi,\la} = (F\times_{\xi,\la} G)_{\xi,\la}$.
  \item If\,\footnote{In \cite{BuchholzSummers:2008,BuchholzLechnerSummers:2010}, this statement is shown to hold under the weaker assumption that the commutator $[\tau_{\xi,s}(F),\,G]$ vanishes only for all $s\in S+S$, where $S$ is the joint spectrum of the generators of the $\Rl^2$-representation $U_\xi$ implementing $\tau_\xi$. But since usually $S=\Rl^2$ in the present setting, we refer here only to the weaker statement, where $S+S\subset\Rl^2$ has been replaced by $\Rl^2$.} $[\tau_{\xi,s}(F),\,G]=0$ for all $s\in\Rl^2$, then $[F_{\xi,\la}, G_{\xi,-\la}]=0$.
  \item If a unitary $Y\in\B(\Hil)$ commutes with $U_\xi(s)$, $s\in\Rl^2$, then $YF_{\xi,\la}Y^{-1} = (YFY^{-1})_{\xi,\la}$, and $YF_{\xi,\la}Y^{-1}$
 is $\xi$-smooth.
 \end{enumerate}
\end{lemma}

Since we are dealing here with a field algebra obeying twisted locality, we also point out that statement c) of the above lemma carries over to the twisted local case.

\begin{lemma}\label{lemma:twist}
Let $\xi\in\Xi$ and $F,G\in\fF$ be $\xi$-smooth such that $[Z \tau_{\xi,s}(F)Z^*,\,G]=0$. Then
\begin{align}
 [ZF_{\xi,\la}Z^*, G_{\xi,-\la}]=0\,.
\end{align}
\end{lemma}
\begin{proof}
The Bose/Fermi operator $V$ commutes with the representation of the isometry group, and thus the same holds true for the twist $Z$ \eqref{def:Z}. So
in view of Lemma \ref{thm:deformationproperties} d), the assumption implies that $ZFZ^*$ is $\xi$-smooth, and $[\tau_{\xi,s}(ZFZ^*),\,G]=0$ for all $s\in\Rl^2$. In view of Lemma \ref{thm:deformationproperties} c), we thus have $[(ZFZ^*)_{\xi,\la}, G_{\xi,-\la}]=0$.  But as $Z$ and $U_\xi(s)$ commute, $(ZFZ^*)_{\xi,\la}=ZF_{\xi,\la}Z^*$, and the claim follows.
\end{proof}

The results summarized in Lemma \ref{thm:deformationproperties} and Lemma \ref{lemma:twist} will be essential for establishing the isotony and twisted locality properties of the deformed quantum field theory. To also control the covariance properties relating different Killing pairs, we need an additional lemma, closely related to \cite[Prop. 2.9]{BuchholzLechnerSummers:2010}.

\begin{lemma}\label{lemma:deformed-operators}
 Let $\xi\in\Xi$, $\la\in\Rl$, and $F\in\fF$ be $\xi$-smooth.
\begin{enumerate}
 \item Let $h\in\iso$. Then $\alpha_h(F)$ is $h_*\xi$-smooth, and
\begin{align}\label{axi-covariance}
 \alpha_h(F_{\xi,\la})
=
\alpha_h(F)_{h_*\xi,\la}
\,.
\end{align}
\item For $N\in\GL(2,\Rl)$, it holds
\begin{align}
 F_{N\xi,\la} &=F_{\xi,\det N\cdot \la}\,. \label{ANxi}
\end{align}
In particular,
\begin{align}\label{Axi-prime}
F_{\xi',\la} &= F_{\xi,-\la}\,.
\end{align}
\end{enumerate}
\end{lemma}
\begin{proof}
a) The flow of $\xi$ transforms under $h$ according to $ h\flow_{\xi,s} = \flow_{h_*\xi,s}h$, so that $\alpha_h(\tau_{\xi,s}(F))=\tau_{h_*\xi,s}(\alpha_h(F))$. Since $F$ is $\xi$-smooth, and $\alpha_h$ is isometric, the smoothness of $s\mapsto\tau_{h_*\xi,s}(\alpha_h(F))$ follows. Using the strong convergence of the oscillatory integrals \eqref{def:Ala}, we compute on a smooth vector $\Psi\in\Hil^\infty$
\begin{align*}
 \alpha_h(F_{\xi,\la})\Psi
&=
\frac{1}{4\pi^2}
\lim_{\eps\to0}
 \int ds\,ds'\, e^{-iss'}\, \chi(\eps s,\eps s')\,U(h \flow_{\xi,\la Qs}h^{-1})\alpha_h(F)U(h \flow_{\xi,s'-\la Qs}h^{-1})\Psi
\\
&=
\frac{1}{4\pi^2}
\lim_{\eps\to0}
 \int ds\,ds'\, e^{-iss'}\, \chi(\eps s,\eps s')\,U(\flow_{h_*\xi,\la Qs})\alpha_h(F)U(\flow_{h_*\xi,s'-\la Qs})\Psi
\\
&=
\alpha_h(F)_{h_*\xi,\la}\Psi
\,,
\end{align*}
which entails \eqref{axi-covariance} since $\Hil^\infty\subset\Hil$ is dense.

b) In view of the transformation law $\varphi_{N\xi,s}=\varphi_{\xi,N^Ts}$ \eqref{N-flow}, we get, $\Psi\in\Hil^\infty$,
\begin{align*}
&F_{N\xi,\la}\Psi
=
\frac{1}{4\pi^2}
\lim_{\eps\to0}
 \int ds\,ds'\, e^{-iss'}\, \chi(\eps s,\eps s')\,U(\flow_{N\xi,\la Qs}) F U(\flow_{N\xi,s'-\la Qs})\Psi
\\
&=
\frac{1}{4\pi^2|\det N|}
\lim_{\eps\to0}
 \int ds\,ds'\, e^{-i(N^{-1}s,s')}\, \chi(\eps s,\eps (N^T)^{-1}s')\,
U_\xi(\la N^TQs) F U_\xi(s'-\la N^TQs)\Psi
\\
&=
\frac{1}{4\pi^2}
\lim_{\eps\to0}
 \int ds\,ds'\, e^{-iss'}\, \chi(\eps N s,\eps (N^T)^{-1}s')\,
U_\xi(\la N^TQNs) F U_\xi(s'-\la N^TQN s)\Psi
\\
&=
F_{\xi,\det N\cdot \la}\Psi
\,.
\end{align*}
In the last line, we used the fact that the value of the oscillatory integral does not depend on the choice of cutoff function $\chi$ or $\chi_N(s,s'):=\chi(Ns,(N^T)^{-1}s')$, and the equation $N^TQN=\det N\cdot Q$, which holds for any $(2\times2)$-matrix $N$.

This proves \eqref{ANxi}, and since $\xi'=\Pi\xi$, with the flip matrix $\Pi=\left(0\;\;1\atop1\;\;0\right)$ which has $\det\Pi=-1$, also \eqref{Axi-prime} follows.
\end{proof}

Having established these properties of individual deformed operators, we now set out to deform the net $W\mapsto\fF(W)$ of wedge algebras. In contrast to the Minkowski space setting \cite{BuchholzLechnerSummers:2010}, we are here in a situation where the set $\Xi$ of all Killing pairs is not a single orbit of one reference pair under the isometry group. Whereas the deformation of a net of wedge algebras on Minkowski space amounts to deforming a single algebra associated with a fixed reference wedge (Borchers triple), we have to specify here more data, related to the coherent subfamilies $\W_{[\xi]}$ in the decomposition $\W=\bigsqcup_{[\xi]}\W_{[\xi]}$ of $\W$ \eqref{W-decomposition}. For each equivalence class $[\xi]$, we choose a representative $\xi$. In case there exists only a single equivalence class, this simply amounts to fixing a reference wedge together with a length scale for the Killing flow. With this choice  of representatives $\xi\in[\xi]$ made, we introduce the sets, $p\in M$,
\begin{align}
\fF(W_{\xi,p})_\la
&:=
\{F_{\xi,\la}\,:\, F\in\fF(W_{\xi,p}) \;\,\xi\text{-smooth } \}^{\|\cdot\|}
\,,
\label{def:awla}
\\
\fF(W_{\xi',p})_\la
&:=
\{F_{\xi',\la}\,:\, F\in\fF(W_{\xi,p}') \;\,\xi'\text{-smooth } \}^{\|\cdot\|}
\,.
\label{def:awla-2}
\end{align}
Here $\la\in\Rl$ is the deformation parameter, and the superscript denotes norm closure in $\B(\Hil)$. Note that the deformed operators in $\fF(W_{\xi',p})_\la$ have the form $F_{\xi',\la}=F_{\xi,-\la}$ \eqref{Axi-prime}, {\em i.e.}, the sign of the deformation parameter depends on the choice of reference Killing pair.

The definitions (\ref{def:awla}, \ref{def:awla-2}) are extended to arbitrary wedges by setting
\begin{align}\label{def:AhWla}
 \fF(hW_{\xi,p})_\la := \alpha_h(\fF(W_{\xi,p})_\la)
\,,\qquad
 \fF(hW_{\xi,p}')_\la := \alpha_h(\fF(W_{\xi,p}')_\la)\,.
\end{align}
Recall that as $h$, $p$ and $[\xi]$ vary over $\iso$, $M$ and $\Xi\slash\!\!\sim$, respectively, this defines $\fF(W)_\la$ for all $W\in\W$ ({\it cf.} \eqref{def:Wxi}). It has to be proven that this assignment is well-defined, {\em e.g.} that \eqref{def:AhWla} is independent of the way the wedge $hW_{\xi,p}=\hti W_{\xi,\pti}$ is represented. This will be done below. However, note that the definition of $\fF(W)_\la$ {\em does} depend on our choice of representatives $\xi\in[\xi]$, since rescaling $\xi$ amounts to rescaling the deformation parameter (Lemma \ref{lemma:deformed-operators} b)).
\\
\\
Before establishing the main properties of the assignment $W\to\fF(W)_\la$, we check that the sets (\ref{def:awla}, \ref{def:awla-2}) are $C^*$-algebras. As the $C^*$-algebra $\fF(W_{\xi,p})$ is $\tau_{\xi}$-invariant and $\tau_{\xi}$ acts strongly continuously, the $\xi$-smooth operators in $\fF(W_{\xi,p})$ which appear in the definition \eqref{def:awla} form a norm-dense ${}^*$-subalgebra. Now the deformation $F \mapsto F_{\xi,\la}$ is evidently linear and commutes with taking adjoints (Lemma \ref{thm:deformationproperties} a)); so the sets $\fF(W_{\xi,p})_\la$ are ${}^*$-invariant norm-closed subspaces of $\B(\Hil)$. To check that these spaces are also closed under taking products, we again use the invariance of $\fF(W_{\xi,p})$ under $\tau_\xi$: By inspection of the Rieffel product \eqref{rieffel-product}, it follows that for any two $\xi$-smooth $F,G\in\fF(W_{\xi,p})$,  also the product $F\times_{\xi,\la} G$ lies in this algebra (and is $\xi$-smooth, see \cite{Rieffel:1992}). Hence the multiplication formula from Lemma \ref{thm:deformationproperties} b) entails that the above defined $\fF(W)_\la$ are actually $C^*$-algebras in $\B(\Hil)$.
\\
\\
The map $W\mapsto\fF(W)_\la$ defines the wedge-local field algebras of the deformed quantum field theory. Their basic properties are collected in the following theorem.

\begin{theorem}\label{thm:deformed-net}
The above constructed map $W\longmapsto\fF(W)_\la$, $W\in\W$, is a well-defined, isotonous, twisted wedge-local, $\iso$-covariant net of $C^*$-algebras on $\Hil$, i.e., $W,\Wti\in\W$,
\begin{align}
& \fF(W)_\la \subset \fF(\Wti)_\la \hspace*{19mm}\text{for } W\subset\Wti\,,\label{isotony2}
\\
& [Z F_\la Z^*,\, G_\la] = 0\hspace*{19mm}\text{for } F_\la\in\fF(W)_\la, G_\la\in\fF(W')_\la\,,\label{locality2}
\\
& \alpha_h(\fF(W)_\la) = \fF(hW)_\la\,,\hspace*{8mm} h\in \iso\,.\label{covariance2}
\end{align}
For $\la=0$, this net coincides with the original net, $\fF(W)_0=\fF(W)$, $W\in\W$.
\end{theorem}
\begin{proof}
It is important to note from the beginning that all claimed properties relate only wedges in the same coherent subfamily $\W_{[\xi]}$. This can be seen from the form \eqref{def:Wxi} of $\W_{[\xi]}$, which is manifestly invariant under isometries and causal complementation, and the structure of the inclusions (Proposition \ref{prop:inclusions}). So in the following proof, it is sufficient to consider a fixed but arbitrary equivalence class $[\xi]$, with selected representative $\xi$.

We begin with establishing the isotony of the deformed net, and therefore consider inclusions of wedges of the form $hW_{\xi,p}, hW_{\xi,p}'$, with $h\in\iso$, $p\in M$ arbitrary, and $\xi\in\Xi$ fixed. Starting with the inclusions  $hW_{\xi,p}\subseteq\hti W_{\xi,\pti}$, we note that according to \eqref{h-W} and Prop. \ref{prop:inclusions}, there exists $N\in\GL(2,\Rl)$ with positive determinant such that $h_*\xi=N\hti_*\xi$. Equivalently, $(\hti^{-1}h)_*\xi=N\xi$, which by Lemma \ref{lemma:group-xi} implies $\det N=1$. By definition, a generic $\xi$-smooth element of $\fF(hW_{\xi,p})_\lambda$ is of the form $\alpha_h(F_{\xi,\la})=\alpha_h(F)_{h_*\xi,\la}$ with some $\xi$-smooth $F\in\fF(W_{\xi,p})$. But according to the above observation, this can be rewritten as
\begin{align}\label{calc1}
\alpha_h(F_{\xi,\la})
=
\alpha_h(F)_{h_*\xi,\la}
=
\alpha_h(F)_{N\hti_*\xi,\la}
=
\alpha_h(F)_{\hti_*\xi,\la}
\,,
\end{align}
where in the last equation we used $\det N=1$ and Lemma \ref{lemma:deformed-operators} b). Taking into account that $hW_{\xi,p}\subseteq\hti W_{\xi,\pti}$, and that the undeformed net is covariant and isotonous, we have $\alpha_h(F)\in\fF(hW_{\xi,p})\subset\fF(\hti W_{\xi,p})$, and so the very right hand side of \eqref{calc1} is an element of $\fF(\hti W_{\xi,p})_\la$. Going to the norm closures, the inclusion $\fF(h W_{\xi,p})_\la \subset \fF(\hti W_{\xi,p})_\la$ of $C^*$-algebras follows.

Analogously, an inclusion of causal complements, $hW_{\xi,p}'\subseteq\hti W_{\xi,\pti}'$, leads to the inclusion of $C^*$-algebras $\fF(h W_{\xi,p}')_\la \subset \fF(\hti W_{\xi,p}')_\la$, the only difference to the previous argument consisting in an exchange of $h$,$\hti$ and $p,\pti$.

To complete the investigation of inclusions of wedges in $\W_{[\xi]}$, we must also consider the case $hW_{\xi,p}\subseteq \hti W_{\xi,\pti}'=W_{\hti_*\xi',\pti}$. By the same reasoning as before, there exists a matrix $N$ with $\det N=1$ such that $(\hti^{-1}h)_*\xi=N\xi'=N\Pi\xi$ with the flip matrix $\Pi$. So  $N':=N\Pi$ has determinant $\det N'=-1$, and $h_*\xi = N'\hti_*\xi'$. Using \eqref{Axi-prime}, we find for $\xi$-smooth $F\in\fF(hW_{\xi,p})$,
\begin{align}
\alpha_h(F_{\xi,\la})
=
\alpha_h(F)_{h_*\xi,\la}
=
\alpha_h(F)_{N'\hti_*\xi',\la}
=
\alpha_h(F)_{\hti_*\xi',-\la}
\,.
\end{align}
By isotony and covariance of the undeformed net, this deformed operator is an element of $\fF(\hti W_{\xi,\pti}')_\la$ \eqref{def:awla-2}, and taking the norm closure in \eqref{def:awla-2} yields $\fF(hW_{\xi,p})_\la \subset\fF(\hti W_{\xi,\pti}')_\la$. So the isotony \eqref{isotony2} of the net is established. This implies in particular that the net $\fF_\la$ is well-defined, since in case $hW_{\xi,p}$ equals $\hti W_{\xi,\pti}$ or its causal complement, the same arguments yield the equality of $\fF(hW_{\xi,p})_\la$ and $\fF(\hti W_{\xi,\pti})_\la$ respectively $\fF(\hti W_{\xi,\pti}')_\la$.

The covariance of $W\mapsto\fF(W)_\la$ is evident from the definition. To check twisted locality, it is thus sufficient to consider the pair of wedges $W_{\xi,p}$, $W_{\xi,p}'$. In view of the definition of the $C^*$-algebras $\fF(W)_\la$ \eqref{def:awla} as norm closures of algebras of deformed smooth operators, it suffices to show that any $\xi$-smooth $F\in\fF(W_{\xi,p})$, $G\in\fF(W_{\xi',\pti})$ fulfill the commutation relation
\begin{align}\label{comm-thm}
 [Z F_{\xi,\la}Z^*,\,G_{\xi',\la}]=0\,.
\end{align}
But $G_{\xi',\la}=G_{\xi,-\la}$ \eqref{Axi-prime}, and since the undeformed net is twisted local and covariant, we have $[\tau_{\xi,s}(F),\,G]=0$, for all $s\in\Rl^2$, which implies $[Z F_{\xi,\la}Z^*,\,G_{\xi,-\la}]=0$ by Lemma \ref{lemma:twist}.

The fact that setting $\la=0$ reproduces the undeformed net is a straightforward consequence of $F_{\xi,0}=F$ for any $\xi$-smooth operator, $\xi\in\Xi$.
\end{proof}

Theorem \ref{thm:deformed-net} is our main result concerning the structure of deformed quantum field theories on admissible spacetimes: It states that the same covariance and localization properties as on flat spacetime can be maintained in the curved setting. Whereas the action of the isometry group and the chosen representation space of $\fF$ are the same for all values of the deformation parameter $\la$, the concrete $C^*$-algebras $\fF(W)_\la$ depend in a non-trivial and continuous way on $\la$: For a fixed wedge $W$, the collection $\{\fF(W)_\la\,:\,\la\in\Rl\}$ forms a continuous field of $C^*$-algebras \cite{Dixmier:1977}; this follows from Rieffel's results \cite{Rieffel:1992} and the fact that $\fF(W)_\la$ forms a faithful representation of Rieffel's deformed $C^*$-algebra $(\fF(W),\times_\la)$ \cite{BuchholzLechnerSummers:2010}.

For deformed nets on Minkowski space, there also exist proofs showing that the {\em net} $W\mapsto\fF(W)_\la$ depends on $\la$, for example by working in a vacuum representation and calculating the corresponding collision operators \cite{BuchholzSummers:2008}. There one finds as a striking effect of the deformation that the interaction depends on $\la$, {\em i.e.} that deformations of interaction-free models have non-trivial S-matrices. However, on generic curved spacetimes, a distinguished invariant state like the vacuum state with its positive energy representation of the translations does not exist. Consequently, the result concerning scattering theory cannot be reproduced here. Instead we will establish the non-equivalence of the undeformed net $W\mapsto \fF(W)$ and the deformed net $W\mapsto \fF(W)_\la$, $\la\neq0$, in a concrete example model in Section \ref{sec:dirac}.

As mentioned earlier, the family of wedge regions $\W(M,g)$ is causally separating in a subclass of admissible spacetimes, including the Friedmann-Robertson-Walker universes. In this  case, the extension of the net $\fF_\la$ to double cones or similar regions $\OO\subset M$ via
\begin{align}
 \fF(\OO)_\la := \bigcap_{W\supset\OO}\fF(W)_\la
\end{align}
is still twisted local. These algebras contain all operators localized in the region $\OO$ in the deformed theory. On other spacetimes $(M,g)$, such an extension is possible only for special regions, intersections of wedges, whose shape and size depend on the structure of $\W(M,g)$.

 Because of the relation of warped convolution to noncommutative spaces, where sharp localization is impossible, it is expected that $\fF(\OO)$ contains only multiples of the identity if $\OO$ has compact closure. We will study this question in the context of the deformed Dirac field in Section \ref{sec:dirac}.
\\\\
We conclude this section with a remark concerning the relation between the field and observable net structure of deformed quantum field theories. The field net $\fF$ is composed of Bose and Fermi fields, and therefore contains observable as well as unobservable quantities. The former give rise to the {\em observable net} $\frA$ which consists of the subalgebras invariant under the grading automorphism $\gamma$. In terms of the projection $v(F):=\frac{1}{2}(F+\gamma(F))$, the observable wedge algebras are
\begin{align}\label{def:AW}
 \frA(W) := \{F\in\fF(W)\,:\, F=\gamma(F)\} = v(\fF(W))
\,,\qquad
W\in\W\,,
\end{align}
so that $\frA(W)$, $\frA(\Wti)$ commute (without twist) if $W$ and $\Wti$ are spacelike separated.

Since the observables are the physically relevant objects, we could have considered a deformation $\frA(W)\to\frA(W)_\la$ of the observable wedge algebras along the same lines as we did for the field algebras. This approach would have resulted precisely in the $\gamma$-invariant subalgebras of the deformed field algebras $\fF(W)_\la$, {\em i.e.}, the diagram
$$
\begin{CD}
\fF(W)   @>\text{deformation\;}>>  \fF(W)_\la\\
@VvVV  @VVvV\\
\frA(W)  @>\text{deformation\;}>>  \frA(W)_\la
\end{CD}
$$
commutes. This claim can quickly be verified by noting that the projection $v$ commutes with the deformation map $F\mapsto F_{\xi,\la}$.

\subsection{The Dirac field and its deformation}\label{sec:dirac}

After the model-independent description of deformed quantum field theories carried out in the previous section, we now consider the theory of a free Dirac quantum field as a concrete example model. We first briefly recall the notion of Dirac (co)spinors and the classical Dirac equation, following largely \cite{DappiaggiHackPinamonti:2009, Sanders:2008} and partly \cite{Dimock:1982,FewsterVerch:2002}, where the proofs of all the statements below are presented and an extensive description of the relevant concepts is available. Afterwards, we consider the quantum Dirac field using Araki's self-dual CAR algebra formulation \cite{Araki:1971}.
\\
\\
As before, we work on a fixed but arbitrary admissible spacetime $(M,g)$ in the sense of Definition \ref{admissible} and we fix its orientation. Therefore, as a four-dimensional, time oriented and oriented, globally hyperbolic spacetime, $M$ admits a {\em spin structure} $(SM,\rho)$, consisting of a principle bundle $SM$ over $M$ with structure group $\SL(2,\Cl)$, and a smooth bundle homomorphism $\rho$ projecting $SM$ onto the frame bundle $FM$, which is a principal bundle over $M$ with $SO_0(3,1)$ as structure group. The map $\rho$ preserves base points and is equivariant in the sense that it intertwines the natural right actions $R$ of $\SL(2,\Cl)$ on $SM$ and of $\SO_0(3,1)$ on $FM$, respectively,
\begin{align}
\rho\circ R_{\widetilde\Lambda} = R_\Lambda\circ\rho,\qquad\Lambda\in \SO_0(3,1)\,,
\end{align}
with the covering homomorphism $\Lambda\mapsto\widetilde\Lambda$ from $\SL(2,\Cl)$ to $\SO_0(3,1)$.

Although each spacetime of the type considered here has a spin structure \cite[Thm. 2.1, Lemma 2.1]{DappiaggiHackPinamonti:2009}, this is only unique if the underlying manifold is simply connected \cite{Geroch:1968,Geroch:1970}, {\em i.e.}, if all edges are simply connected in the case of an admissible spacetime. In the following, it is understood that a fixed choice of spin structure has been made.

The main object we shall be interested in is the {\em Dirac bundle}, that is the associated vector bundle
\begin{align}
 DM:= SM\times_T\Cl^4
\end{align}
with the representation $T:=D^{(\frac{1}{2},0)}\oplus D^{(0,\frac{1}{2})}$ of $\SL(2,\Cl)$ on $\Cl^4$. {\em Dirac spinors} $\psi$ are smooth global sections
of $DM$, and the space they span will be denoted $\mathcal{E}(DM)$. The dual bundle $D^*M$ is called the {\em dual Dirac bundle}, and its smooth global sections $\psi'\in\mathcal{E}(D^*M)$ are referred to as {\em Dirac cospinors}.

For the formulation of the Dirac equation, we need two more ingredients. The first are the so-called {\em gamma-matrices}, which are the coefficients of a global tensor $\gamma\in\mathcal{E}(T^*M\otimes DM\otimes D^*M)$ such that $\gamma=\gamma_{aB}^Ae^a\otimes E_A\otimes E^B$. Here $E_A$ and $E^B$ with $A,B=1,...,4$ are four global sections of $DM$ and $D^*M$ respectively, such that $(E_A,E^B)=\delta_A^B$, with $(,)$ the natural pairing between dual elements. Notice that $E_A$ descends also from a global section $E$ of $SM$ since we can define $E_A(x):=[(E(x),z_A)]$ where $z_A$ is the standard basis of $\Cl^4$. At the same time, out of $E$, we can construct $e_a$, with $a=1,...,4$, as a set of four global sections of $TM$ once we define $e:=\rho\circ E$ as a global section of the frame bundle, which, in turn, can be read as a vector bundle over $TM$ with $\Rl^4$ as typical fibre. The set of all $e_a$ is often referred to as the non-holonomic basis of the base manifold. In this case upper indices are defined via the natural pairing over $\Rl^4$, that is $(e^b,e_a)=\delta_a^b$. Furthermore we choose the gamma-matrices to be of the following form:
\begin{align}
\gamma_0=\left(\begin{array}{cc}
I_2 & 0 \\
0 & -I_2
\end{array}\right),\qquad \gamma_a=\left(\begin{array}{cc}
0 & \sigma_a \\
-\sigma_a & 0
\end{array}\right),
\qquad a=1,2,3,
\end{align}
where the $\sigma_a$ are the Pauli matrices, and $I_n$ will henceforth denote the $(n\times n)$ identity matrix. These matrices fulfil the anticommutation relation $\left\{\gamma_a,\gamma_b\right\}=2\eta_{ab}I_4,$ with the flat Minkowski metric $\eta$. They therefore depend on the sign convention in the metric signature, and differ from those introduced in \cite{DappiaggiHackPinamonti:2009}, where a different convention was used.

The last ingredient we need to specify is the {\em covariant derivative} (spin connection) on the space of smooth sections of the Dirac bundle, that is
\begin{align}
\nabla:\mathcal{E}(DM)\to\mathcal{E}(T^*M\otimes DM),
\end{align}
whose action on the sections $E_A$ is given as $\nabla E_A=\sigma_{aA}^B e^a E_B$. The connection coefficients can be expressed as $\sigma^B_{aA}=\frac{1}{4}\Gamma^b_{ad}\gamma_{bC}^B\gamma^{dC}_A$, where $\Gamma^b_{ad}$ are the coefficients arising from the Levi-Civita connection expressed in terms of non-holonomic basis \cite[Lemma 2.2]{DappiaggiHackPinamonti:2009}.
\\\\
We can now introduce the Dirac equation for spinors $\psi\in\mathcal{E}(DM)$ and cospinors  $\psi'\in\mathcal{E}(D^*M)$ as
\begin{align}\label{Diraceq}
D\psi  &:=(-i\gamma^\mu\nabla_\mu+mI)\psi=0\\
D'\psi'&:=(+i\gamma^\mu\nabla_\mu+mI)\psi'=0
\,,
\end{align}
where $m\geq0$ is a constant while $I$ is the identity on the respective space.

The Dirac equation has unique advanced and retarded fundamental solutions \cite{DappiaggiHackPinamonti:2009}: Denoting the smooth and compactly supported sections of the Dirac bundle by $\mathcal{D}(DM)$, there exist two continuous linear maps
$$S^\pm:\mathcal{D}(DM)\to\mathcal{E}(DM),$$
such that $S^\pm D=D S^\pm = I$ and $\supp(S^\pm f)\subseteq J^\pm (\supp(f))$ for all $f\in\mathcal{D}(DM)$, where $J^\pm$ stands for the causal future/past. In the case of cospinors, we shall instead talk about $S_*^\pm:\mathcal{D}(D^*M)\to\mathcal{E}(D^*M)$ and they have the same properties of $S^\pm$, except that $S_*^\pm D'=D'S_*^\pm=I$. In analogy with the theory of real scalar fields, the {\em causal propagators} for Dirac spinors and cospinors are defined as $S := S^+-S^-$ and $S_*:=S^+_*-S^-_*$, respectively.

For the formulation of a quantized Dirac field, it is advantageous to collect spinors and cospinors in a single object. We therefore introduce the space
\begin{align}
 \DD := \DD(DM\oplus D^*M)\,,
\end{align}
on which we have the conjugation
\begin{align}
 \Gamma(f_1\oplus f_2) := f_1^*\beta \oplus \beta^{-1}f_2^*\,,
\end{align}
defined in terms of the adjoint $f\mapsto f^*$ on $\Cl^4$ and the Dirac conjugation matrix $\beta$. This matrix is the unique selfadjoint element of $\SL(4,\Cl)$ with the properties that $\gamma^*_a=-\beta\gamma_a\beta^{-1}$, $a=0,...,3$, and $i\beta n^a\gamma_a$ is a positive definite matrix, for any timelike future-directed vector $n$.

Due to these properties, the sesquilinear form on $\DD$ defined as
\begin{equation}\label{sesquilinear}
(f_1 \oplus f_2,\, h_1\oplus h_2)
:=
- i\langle f_1^*\beta, Sh_1\rangle
+ i\langle S_* h_2, \beta^{-1}f^*_2\rangle,
\end{equation}
where $\langle,\rangle$ is the global pairing between $\E(DM)$ and $\DD(D^*M)$ or between $\E(D^*M)$ and $\DD(DM)$, is positive semi-definite. Therefore the quotient
\begin{align}
 \K := \DD/(\ker S\oplus \ker S_*)\,.
\end{align}
has the structure of a pre-Hilbert space, and we denote the corresponding scalar product and norm by $\langle\,.\,,\,.\,\rangle_S$ and $\|f\|_S:=\langle f,f\rangle_S^{1/2}$. The conjugation $\Gamma$ descends to the quotient $\K$, and we denote its action on $\K$ by the same symbol. Moreover, $\Gamma$ is compatible with the sesquilinear form \eqref{sesquilinear} in such a way that it extends to an antiunitary involution $\Gamma=\Gamma^*=\Gamma^{-1}$ on the Hilbert space $\overline{\K}$ \cite[Lemma 4.2.4]{Sanders:2008}.


Regarding covariance, the isometry group $\iso$ of $(M,g)$ naturally acts on the sections in $\DD$ by pullback. In view of the geometrical nature of the causal propagator, this action descends to the quotient $\K$ and extends to a unitary representation $u$ of $\iso$ on the Hilbert space $\overline{\K}$.


Given the pre-Hilbert space $\K$, the conjugation $\Gamma$, and the representation $u$ as above, the quantized Dirac field can be conveniently described as follows \cite{Araki:1971}.  We consider the $C^*$-algebra $\CAR(\K,\Gamma)$, that is, the unique unital $C^*$-algebra generated by the symbols $B(f)$, $f\in\K$, such that, for all $f,h\in\K$,
\begin{enumerate}
\item $f \longmapsto B(f)$ is complex linear,
\item $B(f)^* = B(\Gamma f)$,
\item $\{B(f),\,B(h)\} = \langle\Gamma f,h\rangle_S\cdot 1$.
\end{enumerate}
The field equation is implicit here since we took the quotient with respect to the kernels of $S,S_*$. The standard picture of spinors and cospinors can be recovered via the identifications $\psi(h)=B(0\oplus h)$ and $\psi^\dagger(f)=B(f\oplus 0)$.

As is well known, the Dirac field satisfies the standard assumptions of quantum field theory, and we briefly point out how this model fits into the general framework used in Section \ref{sec:generaldeformations}. The global field algebra $\fF:=\CAR(\K,\Gamma)$ carries a natural $\iso$-action $\alpha$ by Bogoliubov transformations: Since the unitaries $u(h)$, $h\in\iso$, commute with the conjugation $\Gamma$, the maps
\begin{align*}
 \alpha_h(B(f)) := B(u(h)f)\,,\qquad f\in\K,
\end{align*}
extend to automorphisms of $\fF$. Similarly, the grading automorphism $\gamma$ is fixed by
\begin{align*}
 \gamma(B(f)) := -B(f)\,,
\end{align*}
and clearly commutes with $\alpha_h$. The field algebra is faithfully represented on the Fermi Fock space $\Hil$ over $\overline{\K}$, where the field operators take the form
\begin{align}
 B(f) = \frac{1}{\sqrt{2}}(a(f)^*+a(\Gamma f))\,,
\end{align}
with the usual Fermi Fock space creation/annihilation operators $a^{\#}(f)$, $f\in\K$. In this representation, the second quantization $U$ of $u$ implements the action $\alpha$. The Bose/Fermi-grading can be implemented by $V=(-1)^N$, where $N\in\B(\Hil)$ is the Fock space number operator \cite{Foit:1983}.

Regarding the regularity assumptions on the symmetries, recall that the anticommutation relations of the CAR algebra imply that $\|a(f)\|=\|f\|_S$, and thus $f\mapsto B(f)$ is a linear continuous map from $\K$ to $\fF$. As $s\mapsto u_\xi(s)f$ is smooth in the norm topology of $\K$ for any $\xi\in\Xi$, $f\in\K$, this implies that the field operators $B(f)$ transform smoothly under the action $\alpha$. Furthermore, the unitarity of $u$ yields strong continuity of $\alpha$ on all of $\fF$, as required in Section \ref{sec:generaldeformations}.

The field algebra $\fF(W)\subset\fF$ associated with a wedge $W\in\W$ is defined as the unital $C^*$-algebra generated by all $B(f)$, $f\in\K(W)$, where $\K(W)$ is the set of (equivalence classes of) smooth and compactly supported sections of $DM\oplus D^*M$ with support in $W$. Since $\langle \Gamma f,g\rangle_S=0$ for $f\in\K(W)$, $g\in\K(W')$, we have $\{B(f),B(g)\}=0$ for $f\in\K(W)$, $g\in\K(W')$, which implies the twisted locality condition \eqref{twisted-locality}. Isotony is clear from the definition, and covariance under the isometry group follows from $u(h)\K(W)=\K(hW)$.
\\\\
The model of the Dirac field therefore fits into the framework of Section \ref{sec:generaldeformations}, and the warped convolution deformation defines a one-parameter family of deformed nets $\fF_\la$. Besides the properties which were established in the general setting in Section \ref{sec:generaldeformations}, we can here consider the explicit deformed field $B(f)_{\xi,\la}$. A nice characterization of these operators can be given in terms of their $n$-point functions associated with a quasifree state $\om$ on $\fF$.

Let $\om$ be an $\alpha$-invariant quasifree state on $\fF$, and let ($\Hil^\om,\pi^\om,\Om^\om)$ denote the associated GNS triple. As a consequence of invariance of $\om$, the GNS space $\Hil^\om$ carries a unitary representation $U^\om$ of $\iso$ which leaves $\Om^\om$ invariant. In this situation, the warping map
\begin{align}\label{warp-om}
F_{\xi,\la} \mapsto F_{\xi,\la}^\om
:=
\frac{1}{4\pi^2}\lim_{\eps\to0}
\int ds\,ds'\,
e^{-iss'}\,
\chi(\eps s,\eps s')\,
U^\om_\xi(\la Q s)\pi^\om(F)U_\xi^\om(s'-\la Qs)
\,,
\end{align}
defined for $\xi$-smooth $F\in\fF$ as before, extends continuously to a representation of the Rieffel-deformed $C^*$-algebra $(\fF,\times_{\xi,\la})$ on $\Hil^\om$ \cite[Thm. 2.8]{BuchholzLechnerSummers:2010}. Moreover, the $U^\om$-invariance of $\Om^\om$ implies
\begin{align}\label{FxiOm}
F^\om_{\xi,\la}\Om^\om
=
\pi^\om(F)\Om^\om
\,,\qquad \xi\in\Xi,\la\in\Rl,\,F\in\fF\;\;\xi\text{-smooth}.
\end{align}
Since the CAR algebra is simple, all its representations are faithful \cite{BratteliRobinson:1997}. We will therefore identify $\fF$ with its representation $\pi^\om(\fF)$ in the following, and drop the $\om$-dependence of $\Hil^\om,\Om^\om,U^\om$ and the warped convolutions $F^\om_{\xi,\la}$ from our notation.

To characterize the deformed field operators $B(f)_{\xi,\la}$, we will consider the $n$-point functions
\begin{align*}
\om_n(f_1,...\,,f_n)
:=
\om(B(f_1)\cdots B(f_n))
=
\langle\Om,\,B(f_1)\cdots B(f_n)\Om\rangle
\,,\qquad
f_1,...\,,f_n\in\K\,,
\end{align*}
and the corresponding deformed expectation values of the deformed fields, called {\em deformed $n$-point functions},
\begin{align*}
\om^{\xi,\la}_n(f_1,...\,,f_n)
:=
\langle\Om,\,B(f_1)_{\xi,\la}\cdots B(f_n)_{\xi,\la}\Om\rangle
\,,\qquad
f_1,...\,,f_n\in\K\,.
\end{align*}
Of particular interest are the {\em quasifree} states, where $\om_n$ vanishes if $n$ is odd, and  $\om_n$ is a linear combination of products of two-point functions $\om_2$ if $n$ is even. In particular, the undeformed four-point function of a quasifree state reads
\begin{align}\label{4pt-null}
\!\!\!\!\!\!
\om_4(f_1,f_2,f_3,f_4)
=
\om_2(f_1,f_2)\om_2(f_3,f_4)
+\om_2(f_1,f_4)\om_2(f_2,f_3)
-\om_2(f_1,f_3)\om_2(f_2,f_4)
\,.
\!\!\!\!
\end{align}

\begin{proposition}\label{prop:n-pt}
 The deformed $n$-point functions of a quasifree and $\iso$-invariant state vanish for odd $n$. The lowest deformed even $n$-point functions are, $f_1,...,f_4\in\K$,
\begin{align}
 \om_2^{\xi,\la}(f_1,f_2) &= \om_2(f_1,f_2)
\,,\label{2pt}
\\
\omega_4^{\xi,\lambda}(f_1,f_2,f_3,f_4)
&=
\om_2(f_1,f_2)\om_2(f_3,f_4)
+\om_2(f_1,f_4)\om_2(f_2,f_3)
\label{4pt}
\\
&
- \frac{1}{4\pi^2}
\lim_{\varepsilon\to 0}
\int ds\,ds'\,
e^{-iss'}
\chi(\eps s,\eps s')
\,
\om_2(f_1,u_\xi(s)f_3)\cdot \om_2(f_2,u_\xi(2\la Q s')f_4)
\,.
\nonumber
\end{align}
\end{proposition}
\begin{proof}
The covariant transformation behaviour of the Dirac field, $U_\xi(s)B(f)U_\xi(s)^{-1}=B(u_\xi(s)f)$, the invariance of $\Om$ and the form \eqref{warp-om} of the warped convolution imply that any deformed $n$-point function can be written as an integral over undeformed $n$-point functions with transformed arguments. As the latter functions vanish for odd $n$, we also have $\om^{\xi,\la}_n=0$ for odd $n$.

Taking into account \eqref{FxiOm}, we obtain for the deformed two-point function
\begin{align*}
\om^{\xi,\la}_2(f_1,f_2)
&=
\langle\Om,\,B(f_1)_{\xi,\la}B(f_2)_{\xi,\la}\Om\rangle
\\
&=
\langle (B(f_1)^*)_{\xi,\la}\Om,\,B(f_2)_{\xi,\la}\Om\rangle
=
\langle B(f_1)^*\Om,\,B(f_2)\Om\rangle
=
\om_2(f_1,f_2)\,,
\end{align*}
proving \eqref{2pt}.

To compute the four-point function, we use again $B(f)_{\xi,\la}\Om=B(f)\Om$ and $U_\xi(s)\Om=\Om$. Inserting the definition of the warped convolution \eqref{def:Ala}, and using the transformation law $U_\xi(s)B(f)U_\xi(s)^{-1}=B(u_\xi(s)f)$ and the shorthand $f(s):=u_\xi(s)f$, we obtain
\begin{align*}
&\om^{\xi,\la}_4(f_1,f_2,f_3,f_4)
=
\langle\Om,\,B(f_1)B(f_2)_{\xi,\la}B(f_3)_{\xi,\la}B(f_4)\Om\rangle
\\
&=
(2\pi)^{-4}
\lim_{\eps_1,\eps_2\to0}
\int d\bs\,
e^{-i(s_1s_1'+s_2s_2')}
\chi_\eps(\bs)
\,
\om_4(f_1,f_2(\la Qs_1),f_3(\la Q s_2+s_1'),f_4(s_1'+s_2'))\,,
\end{align*}
where $d\bs=ds_1\,ds_1'\,ds_2\,ds_2'$ and $\chi_\eps(\bs)=\chi(\eps_1 s_1,\eps_1 s_1')
\chi(\eps_2 s_2,\eps_2 s_2')$. After the substitutions $s_2'\to s_2'-s_1'$ and $s_1'\to s_1'-\la Q s_2$, the integrations over $s_2,s_2'$ and the limit $\eps_2\to0$ can be carried out. The result is
\begin{align*}
(2\pi)^{-2}
\lim_{\eps_1\to0}
\int ds_1\,ds'_1\,
e^{-is_1s_1'}
\hat{\chi}(\eps_1 s_1,\eps_1 s_1')
\,
\om_4(f_1,f_2(\la Qs_1),f_3(s_1'),f_4(s_1'-\la Q s_1))\,,
\end{align*}
with a smooth, compactly supported cutoff function $\hat{\chi}$ with $\hat{\chi}(0,0)=1$.

We now use the fact that $\om$ is quasi-free and write $\om_4$ as a sum of products of two-point functions \eqref{4pt-null}. Considering the term where $f_1,f_2$  and $f_3,f_4$ are contracted, in the second factor $\om_2(f_3(s'_1),f_4(s_1'-\la Qs_1))$ the $s_1'$-dependence drops out because $\om$ is invariant under isometries. So the integral over $s_1'$ can be performed, and yields a factor $\delta(s_1)$ in the limit $\eps_1\to0$. Hence all $\la$-dependence drops out in this term, as claimed in \eqref{4pt}.

Similarly, in the term where $f_1,f_4$  and $f_2,f_3$ are contracted, all integrations disappear after using the invariance of $\om$ and making the substitution $s_1'\to s_1'+\la Qs_1$. Also this term does not depend on $\la$.

Finally, we compute the term containing the contractions $f_1,f_3$ and $f_2,f_4$,
\begin{align*}
&(2\pi)^{-2}
\lim_{\eps_1\to0}
\int ds_1\,ds'_1\,
e^{-is_1s_1'}
\hat{\chi}(\eps_1 s_1,\eps_1 s_1')
\,
\om_2(f_1,f_3(s_1'))\cdot \om_2(f_2(\la Qs_1),f_4(s_1'-\la Q s_1))
\\
&=
(2\pi)^{-2}
\lim_{\eps_1\to0}
\int ds_1\,ds'_1\,
e^{-is_1s_1'}
\hat{\chi}(\eps_1 s_1,\eps_1 s_1')
\,
\om_2(f_1,f_3(s_1'))\cdot \om_2(f_2,f_4(s_1'-2\la Q s_1))
\\
&=
(2\pi)^{-2}
\lim_{\eps_1\to0}
\int ds_1\,ds'_1\,
e^{-is_1s_1'}
\tilde{\chi}(\eps_1 s_1,\eps_1 s_1')
\,
\om_2(f_1,f_3(s_1'))\cdot \om_2(f_2,f_4(2\la Q s_1))
\,.
\end{align*}
In the last step, we substituted $s_1\to s_1+\frac{1}{2\la}Q^{-1}s_1'$, used the antisymmetry of $Q$ and absorbed the new variables in $\hat{\chi}$ in a redefinition of this function. Since the oscillatory integrals are independent of the particular choice of cutoff function, comparison with \eqref{4pt} shows that the proof is finished.
\end{proof}

The structure of the deformed $n$-point functions build from a quasifree state $\om$ is quite different from the undeformed case. In particular, the two-point function is undeformed, but the four-point function depends on the deformation parameter. For even $n>4$, a structure similar to the $n$-point functions on noncommutative Minkowski space \cite{GrosseLechner:2008} is expected, which are all $\lambda$-dependent. These features clearly show that the deformed field $B(f)_{\xi,\la}$, $\la\neq0$, differs from the undeformed field. Considering the commutation relations of the deformed field operators, it is also straightforward to check that the deformed field is not unitarily equivalent to the undeformed one.

However, this structure does not yet imply that the deformed and undeformed Dirac quantum field theories are inequivalent. For there could exist a unitary $V$ on $\Hil$ satisfying $VU(h)V^*=U(h)$, $h\in\iso$, $V\Om=\Om$, which does not interpolate the deformed and undeformed fields, but the $C^*$-algebras according to $V\fF(W)_\la V^*=\fF(W)$, $W\in\W$. If such a unitary exists, the two theories would be physically indistinguishable.

On flat spacetime, an indirect way of ruling out the existence of such an intertwiner $V$, and thus establishing the non-equivalence of deformed and undeformed theories, is to compute their S-matrices and show that these depend on $\la$ in a non-trivial manner. However, on curved spacetimes, collision theory is not available and we will therefore follow the more direct non-equivalence proof of \cite[Lemma 4.6]{BuchholzLechnerSummers:2010}, adapted to our setting. This proof aims at showing that the local observable content of warped theories is restricted in comparison to the undeformed setting, as one would expect because of the connection ot noncommutative spacetime. However, the argument requires a certain amount of symmetry, and we therefore restrict here to the case of a Friedmann-Robertson-Walker spacetime $M$.

As discussed in Section \ref{sec:examples}, $M$ can then be viewed as $J\times\Rl^3\subset\Rl^4$ via a conformal embedding, where $J\subset\Rl$ is an interval. Recall that in this case, we have the Euclidean group E$(3)$ contained in $\iso(M,g)$, and can work in global coordinates $(\tau,x,y,z)$, where $\tau\in J$ and $x,y,z\in\Rl$ are the flow parameters of Killing fields. As reference Killing pair, we pick $\zeta:=(\partial_y,\partial_z)$, and as reference wedge, the ``right wedge'' $W^0:=W_{\zeta,0}=\{(\tau,x,y,z)\,:\,\tau\in J,\,x>|\tau|\}$.

In this geometric context, consider the rotation $r^\varphi$ about angle $\varphi$ in the $x$-$y$-plane, and the cone
\begin{align}\label{def:cone}
 \C := r^\varphi W^0 \cap r^{-\varphi} W^0\,,
\end{align}
with some fixed angle $|\varphi|<\frac{\pi}{2}$. Clearly $\C\subset W^0$, and the reflected cone $j_x\C$, where $j_x(t,x,y,z)=(t,-x,y,z)$, lies spacelike to $W^0$ and $r^\varphi W^0$.

Moreover, we will work in the GNS-representation of a particular state $\om$ on $\fF$ for the subsequent proposition, which besides the properties mentioned above also has the Reeh-Schlieder property. That is, the von Neumann algebra $\fF(\C)''\subset\B(\Hil^\om)$ has $\Om^\om$ as a cyclic vector.

Since the Dirac field theory is a locally covariant quantum field theory satisfying the time slice axiom \cite{Sanders:2009-2}, the existence of such states can be deduced by spacetime deformation arguments \cite[Thm. 4.1]{Sanders:2009-1}. As $M$ and the Minkowski space have unique spin structures, and $M$ can be deformed to Minkowski spacetime in such a way that its E$(3)$ symmetry is preserved, the state obtained from deforming the Poincar\'e invariant vacuum state on $\Rl^4$ is still invariant under the action of the Euclidean group.

In the GNS representation of such a Reeh-Schlieder state on a Friedmann-Robertson-Walker spacetime, we find the following non-equivalence result.

\begin{proposition}\label{prop:inequivalence}
Consider the net $\fF_\la$ generated by the deformed Dirac field on a Fried\-mann-Robertson-Walker spacetime with flat spatial sections in the GNS representation of an invariant state with the Reeh-Schlieder property. Then the implementing vector $\Om$ is cyclic for the field algebra ${\fF(\C)_\la}''$ associated with the cone \eqref{def:cone} if and only if $\la=0$. In particular, the nets $\fF_0$ and $\fF_\la$ are inequivalent for $\la\neq0$.
\end{proposition}
\begin{proof}
Let $f\in\K(\C)$, so that $f, u(r^{-\varphi})f\in\K(W^0)$, and both field operators, $B(f)_{\zeta,\la}$ and $B(u(r^{-\varphi})f)_{\zeta,\la}$, are contained in $\fF(W^0)_\la$. Taking into account the covariance of the deformed net, it follows that $U(r^\varphi)B(u(r^{-\varphi})f)_{\zeta,\la} U(r^{-\varphi}) = B(f)_{r^\varphi_*\zeta,\la}$ lies in $\fF(r^\varphi W^0)_\la$.

Now the cone $\C$ is defined in such a way that the two wedges $W^0$ and $r^\varphi W^0$ lie spacelike to $j_x\C$. Let us assume that $\Om$ is cyclic for ${\fF(\C)_\la}''$, which by the unitarity of $U(j_x)$ is equivalent to $\Om$ being cyclic for ${\fF(j_x\C)_\la}''$. Hence $\Om$ is separating for the commutant ${\fF(j_x\C)_\la}'$, which by locality contains $\fF(W^0)_\la$ and $\fF(r^\varphi W^0)_\la$. But in view of \eqref{FxiOm}, $B(f)_{\zeta,\la}$ and $B(f)_{r^\varphi_*\zeta,\la}$ coincide on $\Om$,
\begin{align*}
B(f)_{r^\varphi_*\zeta,\la}\Om
=
B(f)\Om
=
B(f)_{\zeta,\la}\Om
\,,
\end{align*}
so that the separation property implies $B(f)_{\zeta,\la}=B(f)_{r^\varphi_*\zeta,\la}$.

To produce a contradiction, we now show that these two operators are actually not equal. To this end, we consider a difference of four-point functions \eqref{4pt}, with smooth vectors $f_1,f_2:=f,f_3:=f,f_4$, and Killing pairs $\zeta$ respectively $r^\varphi_*\zeta$. With the abbreviations $w^\varphi_{ij}(s):=\om_2(f_i,u_{r^\varphi_*\zeta}(s)f_j)$, we obtain
\begin{align*}
&\hspace*{-4cm}
\langle\Om,B(f_1)\left(B(f)_{\zeta,\la}B(f)_{\zeta,\la}
-
B(f)_{r^\varphi_*\zeta,\la}B(f)_{r^\varphi_*\zeta,\la}\right)B(f_4)\Om\rangle
\\
&=
\om^{\zeta,\la}_4(f_1,f,f,f_4)
-
\om^{r^\varphi_*\zeta,\la}_4(f_1,f,f,f_4)
\\
&=
(w^0_{13}\star_\la w^0_{24})(0)
-
(w^\varphi_{13}\star_\la w^\varphi_{24})(0)\,,
\end{align*}
where $\star_\la$ denotes the Weyl-Moyal star product on smooth bounded functions on $\Rl^2$, with the standard Poisson bracket given by the matrix \eqref{def:Q} in the basis $\{\zeta_1,\zeta_2\}$. Now the asymptotic expansion of this expression for $\la\to0$ gives in first order the difference of Poisson brackets \cite{EstradaGracia-BondiaVarilly:1989}
\begin{align*}
\{w^0_{13},\,w^0_{24}\}(0)-\{w^\varphi_{13},\,w^\varphi_{24}\}(0)
&=
\langle f_1, P^\zeta_1 f\rangle\langle f, P^\zeta_2 f_4\rangle
-
\langle f_1, P^\zeta_2 f\rangle\langle f, P^\zeta_1 f_4\rangle
\\
&\qquad -
\langle f_1, P^{r^\varphi_*\zeta}_1 f\rangle\langle f, P^{r^\varphi_*\zeta}_2 f_4\rangle
+
\langle f_1, P^{r^\varphi_*\zeta}_2 f\rangle\langle f, P^{r^\varphi_*\zeta}_1 f_4\rangle
\,,
\end{align*}
where all scalar products are in $\K$ and $P^{r^\varphi_*\zeta}_1,P^{r^\varphi_*\zeta}_1$ denote the generators of $s\mapsto u_{r^\varphi_*\zeta}(s)$. By considering $f_4$ orthogonal to $P_1^\zeta f$ and $P_1^{r^\varphi_*\zeta}f$, we see that for $B(f)_{\zeta,\la}=B(f)_{r^\varphi_*\zeta,\la}$ it is necessary that $\langle f_1,(P^\zeta_j-P^{r^\varphi_*\zeta}_j)f\rangle=0$. But varying $f,f_1$ within the specified limitations gives dense subspaces in $\K$, {\em i.e.} we must have $P^\zeta_j=P^{r^\varphi_*\zeta}_j$. This implies that translations in a spacelike direction are represented trivially on the Dirac field, which is not compatible with its locality and covariance properties.

So we conclude that the deformed field operator $B(f)_{r^\varphi_*\zeta,\la}$ is not independent of $\varphi$ for $\la\neq0$, and hence the cyclicity assumption is not valid for $\la\neq0$. Since on the other hand $\Om$ is cyclic for ${\fF(\C)_0}''$ by the Reeh-Schlieder property of $\om$, and a unitary $V$ leaving $\Om$ invariant and mapping $\fF(\C)_0$ onto $\fF(\C)_\la$ would preserve this property, we have established that the nets $\fF_0$ and $\fF_\la$, $\la\neq0$, are not equivalent.
\end{proof}

\section{Conclusions}\label{sec:outlook}

In this paper we have shown how to apply the warped convolution deformation to quantum field theories on a large class of globally hyperbolic spacetimes. Under the requirement that the group of isometries contains a two-dimensional Abelian subgroup generated by two commuting spacelike Killing fields, many results known for Minkowski space theories were shown to carry over to the curved setting by formulating concepts like edges, wedges and translations in a geometrical language. In particular, it has been demonstrated in a model-independent framework that the basic covariance and wedge-localization properties we started from are preserved for all values of the deformation parameter $\la$. As a concrete example we considered a warped Dirac field on a Friedmann-Robertson-Walker spacetime in the GNS representation of a quasifree and $\iso$-invariant state with Reeh-Schlieder property. It was shown that the deformed models depend in a non-trivial way on $\la$, and violate the Reeh-Schlieder property for regions smaller than wedges. In view of the picture that the deformed models can be regarded as effective theories on a noncommutative spacetime, where strictly local quantities do not exist because of space-time uncertainty relations, it is actually expected that they do not contain operators sharply localized in bounded spacetime regions for $\la\neq0$.

At the current stage, it is difficult to give a clear-cut physical interpretation to the models constructed here since scattering theory is not available for quantum field theories on generic curved spacetimes. Nonetheless, it is interesting to note that in a field theoretic context the deformation often leaves invariant the two-point function (Prop. \ref{prop:n-pt}), the quantity which is most frequently used for deriving observable quantum effects in cosmology (for the example of quantised cosmological perturbations, see \cite{MukhanovFeldmanBrandenberger:1990}). So when searching for concrete scenarios where deformed quantum field theories can be matched to measurable effects, one has to look for phenomena involving the higher $n$-point functions.
\\
\\
There exist a number of interesting directions in which this research could be extended. As far as our geometrical construction of wedge regions in curved spacetimes is concerned, we limited ourselves here to edges generated by commuting Killing fields. This assumption rules out many physically interesting spacetimes, such as de Sitter, Kerr or Friedmann-Robertson-Walker spacetimes with compact spatial sections. An extension of the geometric construction of edges and wedges to such spaces seems to be straightforward and is expected to coincide with the notions which are already available. For example, in the case of four-dimensional de Sitter spacetime, edges have the topology of a two-sphere. Viewing the two-sphere as an $\SO(3)$ orbit, a generalization of the warped convolution deformation formula could involve an integration over this group instead of $\Rl^2$. A deformation scheme of $C^*$-algebras which is based on an action of this group is currently not yet available (see however \cite{Bieliavsky:2007} for certain non-Abelian group action). In the de Sitter case, a different possibility is to base the deformation formula on the flow of a spacelike and commuting timelike Killing field leaving a wedge invariant. The covariance properties of the nets which arise from such a procedure are somewhat different and require an adaption of the kernel in the warped convolution \cite{MorfaMorales:2010}.

\newpage
It is desireable to find more examples of adequate deformation formulas compatible with the basic covariance and (wedge-) localization properties of quantum field theory. A systematic exploration of the space of all deformations is expected to yield many new models and a better understanding of the structure of interacting quantum field theories.

\vspace*{9mm}
{\noindent\bf Acknowledgements}\\
C.D. gratefully acknowledges financial support from the Junior Fellowship Programme of the Erwin Schr\"odinger Institute and from the German Research Foundation DFG through the Emmy Noether Fellowship WO 1447/1-1. 
E.M.M. would like to thank the EU-network MRTN-CT-2006-031962, EU-NCG for financial support.
We are grateful to Valter Moretti and Helmut Rumpf for helpful discussions.

{\footnotesize

\newcommand{\etalchar}[1]{$^{#1}$}

}


\begin{thebibliography}{GGBI{\etalchar{+}}04}
\providecommand{\href}[1]{\texttt{#1}{link}}

\bibitem[ABD{\etalchar{+}}05]{AschieriBlohmannDimitrijevicMeyerSchuppWess:2005}
P.~Aschieri, C.~Blohmann, M.~Dimitrijevic, F.~Meyer, P.~Schupp, and J.~Wess.
\newblock {A Gravity Theory on Noncommutative Spaces}.
\newblock \href{http://dx.doi.org/10.1088/0264-9381/22/17/011}{{\em
  Class.Quant.Grav.} {\bf 22} (2005)  3511--3532},
  \href{http://arxiv.org/abs/hep-th/0504183}{\tt{[open access]}}.

\bibitem[Ara71]{Araki:1971}
H.~Araki.
\newblock {On quasifree states of CAR and Bogoliubov automorphisms}.
\newblock \href{http://dx.doi.org/10.2977/prims/1195193913}{{\em Publ. Res.
  Inst. Math. Sci. Kyoto} {\bf 6} (1971)  385--442},
  \href{http://www.ems-ph.org/journals/show_abstract.php?issn=0034-5318&vol=6&%
iss=3&rank=1&srch=searchterm%7COn+quasifree+states+of+CAR+and+Bogoliubov+autom%
orphisms}{\tt{[open access]}}.

\bibitem[Ara99]{Araki:1999}
H.~Araki.
\newblock {\em {Mathematical Theory of Quantum Fields}}.
\newblock {Int. Series of Monographs on Physics}. Oxford University Press,
  Oxford, 1999.

\bibitem[BB99]{BorchersBuchholz:1999}
H.-J. Borchers and D.~Buchholz.
\newblock {Global properties of vacuum states in de Sitter space}.
\newblock {\em Annales Poincare Phys. Theor.} {\bf A70} (1999)  23--40,
  \href{http://arxiv.org/abs/gr-qc/9803036}{\tt{[open access]}}.

\bibitem[BDFP10]{BahnsDoplicherFredenhagenPiacitelli:2010}
D.~Bahns, S.~Doplicher, K.~Fredenhagen, and G.~Piacitelli.
\newblock
\newblock {Quantum Geometry on Quantum Spacetime: Distance, Area and Volume
  Operators }{\em Preprint} (May, 2010)  ,
  \href{http://arxiv.org/abs/1005.2130}{\tt{[open access]}}.

\bibitem[BDFS00]{BuchholzDreyerFlorigSummers:2000}
D.~Buchholz, O.~Dreyer, M.~Florig, and S.~J. Summers.
\newblock {Geometric modular action and spacetime symmetry groups}.
\newblock {\em Rev. Math. Phys.} {\bf 12} (2000)  475--560,
  \href{http://arxiv.org/abs/math-ph/9805026}{\tt{[open access]}}.

\bibitem[BGK{\etalchar{+}}08]{BlaschkeGieresKronbergerSchwedaWohlgenannt:2008}
D.~Blaschke, F.~Gieres, E.~Kronberger, M.~Schweda, and M.~Wohlgenannt.
\newblock {Translation-invariant models for non-commutative gauge fields}.
\newblock \href{http://dx.doi.org/10.1088/1751-8113/41/25/252002}{{\em
  J.Phys.A} {\bf 41} (2008)  252002},
  \href{http://arxiv.org/abs/0804.1914}{\tt{[open access]}}.

\bibitem[BGL02]{BrunettiGuidoLongo:2002}
R.~Brunetti, D.~Guido, and R.~Longo.
\newblock {Modular localization and Wigner particles}.
\newblock {\em Rev. Math. Phys.} {\bf 14} (2002)  759--786,
  \href{http://arxiv.org/abs/math-ph/0203021}{\tt{[open access]}}.

\bibitem[Bie07]{Bieliavsky:2007}
P.~Bieliavsky.
\newblock
\newblock {Deformation quantization for actions of the affine group }{\em
  Preprint} (September, 2007)  ,
  \href{http://arxiv.org/abs/0709.1110}{\tt{[open access]}}.

\bibitem[BLS10]{BuchholzLechnerSummers:2010}
D.~Buchholz, G.~Lechner, and S.~J. Summers.
\newblock {Warped Convolutions, Rieffel Deformations and the Construction of
  Quantum Field Theories}.
\newblock {\em Preprint} (2010)  ,
  \href{http://arxiv.org/abs/1005.2656}{\tt{[open access]}}.

\bibitem[BMS01]{BuchholzMundSummers:2001}
D.~Buchholz, J.~Mund, and S.~J. Summers.
\newblock {Transplantation of Local Nets and Geometric Modular Action on
  Robertson-Walker Space-Times}.
\newblock {\em Fields Inst. Commun.} {\bf 30} (2001)  65--81,
  \href{http://arxiv.org/abs/hep-th/0011237}{\tt{[open access]}}.

\bibitem[Bor92]{Borchers:1992}
H.-J. Borchers.
\newblock {The CPT theorem in two-dimensional theories of local observables}.
\newblock {\em Commun. Math. Phys.} {\bf 143} (1992)  315--332.

\bibitem[Bor00]{Borchers:2000}
H.-J. Borchers.
\newblock {On revolutionizing quantum field theory with Tomita's modular
  theory}.
\newblock {\em J. Math. Phys.} {\bf 41} (2000)  3604--3673.

\bibitem[Bor09]{Borchers:2009}
H.-J. Borchers.
\newblock {On the Net of von Neumann algebras associated with a Wedge and
  Wedge-causal Manifolds}.
\newblock {\em Preprint} (2009)  ,
  \href{http://www.lqp.uni-goe.de/papers/09/12/09120802.html}{\tt{[open
  access]}}.

\bibitem[BPQV08]{BalachandranPinzulQureshiVaidya:2007}
A.~P. Balachandran, A.~Pinzul, B.~A. Qureshi, and S.~Vaidya.
\newblock {S-Matrix on the Moyal Plane: Locality versus Lorentz Invariance}.
\newblock
  \href{http://dx.doi.org/http://dx.doi.org/10%2E1103/PhysRevD%2E77%2E025020}{%
{\em Phys.Rev.D} {\bf 77} (2008)  025020},
  \href{http://arxiv.org/abs/0708.1379}{\tt{[open access]}}.

\bibitem[BR97]{BratteliRobinson:1997}
O.~Bratteli and D.~W. Robinson.
\newblock {\em {Operator Algebras and Quantum Statistical Mechanics II}}.
\newblock Springer, 1997.

\bibitem[BS04]{BuchholzSummers:2004-2}
D.~Buchholz and S.~J. Summers.
\newblock {Stable quantum systems in anti-de Sitter space: Causality,
  independence and spectral properties}.
\newblock \href{http://dx.doi.org/10.1063/1.1804230}{{\em J. Math. Phys.} {\bf
  45} (2004)  4810--4831},
  \href{http://www.arxiv.org/abs/math-ph/0407011}{\tt{[open access]}}.

\bibitem[BS05]{BernalSanchez:2005}
A.~N. Bernal and M.~S{\'a}nchez.
\newblock {Smoothness of time functions and the metric splitting of globally
  hyperbolic spacetimes}.
\newblock \href{http://dx.doi.org/10.1007/s00220-005-1346-1}{{\em
  Commun.Math.Phys.} {\bf 257} (2005)  43--50},
  \href{http://arxiv.org/abs/gr-qc/0401112}{\tt{[open access]}}.

\bibitem[BS06]{BernalSanchez:2006}
A.~N. Bernal and M.~S{\'a}nchez.
\newblock {Further results on the smoothability of Cauchy hypersurfaces and
  Cauchy time functions}.
\newblock \href{http://dx.doi.org/10.1007/s11005-006-0091-5}{{\em
  Lett.Math.Phys.} {\bf 77} (2006)  183--197},
  \href{http://arxiv.org/abs/gr-qc/0512095}{\tt{[open access]}}.

\bibitem[BS07]{BuchholzSummers:2007}
D.~Buchholz and S.~J. Summers.
\newblock {String- and brane-localized fields in a strongly nonlocal model}.
\newblock \href{http://dx.doi.org/10.1088/1751-8113/40/9/019}{{\em J. Phys.}
  {\bf A40} (2007)  2147--2163},
  \href{http://arxiv.org/abs/math-ph/0512060}{\tt{[open access]}}.

\bibitem[BS08]{BuchholzSummers:2008}
D.~Buchholz and S.~J. Summers.
\newblock {Warped Convolutions: A Novel Tool in the Construction of Quantum
  Field Theories}.
\newblock In E.~Seiler and K.~Sibold, editors, {\em {Quantum Field Theory and
  Beyond: Essays in Honor of Wolfhart Zimmermann}}, pages 107--121. World
  Scientific, 2008.
\newblock \href{http://arxiv.org/abs/0806.0349}{\tt{[open access]}}.

\bibitem[BW75]{BisognanoWichmann:1975}
J.~J Bisognano and E.~H. Wichmann.
\newblock {On the Duality Condition for a Hermitian Scalar Field}.
\newblock {\em J. Math. Phys.} {\bf 16} (1975)  985--1007.

\bibitem[CF84]{ChandrasekharFerrari:1984}
S.~Chandrasekhar and V.~Ferrari.
\newblock {On the Nutku-Halil solution for colliding impulsive gravitational
  waves}.
\newblock {\em Proc. Roy. Soc. Lond.} {\bf A396} (1984)  55.

\bibitem[Cha83]{Chandrasekhar:1983}
S.~Chandrasekhar.
\newblock {\em {The mathematical theory of black holes}}.
\newblock Oxford University Press, 1983.

\bibitem[DFR95]{DoplicherFredenhagenRoberts:1995}
S.~Doplicher, K.~Fredenhagen, and J.~E. Roberts.
\newblock {The Quantum structure of space-time at the {P}lanck scale and
  quantum fields}.
\newblock {\em Commun. Math. Phys.} {\bf 172} (1995)  187--220,
  \href{http://arxiv.org/abs/hep-th/0303037}{\tt{[open access]}}.

\bibitem[DHP09]{DappiaggiHackPinamonti:2009}
C.~Dappiaggi, T.-P. Hack, and N.~Pinamonti.
\newblock {The extended algebra of observables for Dirac fields and the trace
  anomaly of their stress-energy tensor}.
\newblock \href{http://dx.doi.org/10.1142/S0129055X09003864}{{\em
  Rev.Math.Phys.} {\bf 21} (2009)  1241--1312},
  \href{http://arxiv.org/abs/0904.0612}{\tt{[open access]}}.

\bibitem[DHR69]{DoplicherHaagRoberts:1969}
S.~Doplicher, R.~Haag, and J.~E. Roberts.
\newblock {Fields, observables and gauge transformations. I}.
\newblock {\em Commun. Math. Phys.} {\bf 13} (1969)  1--23,
  \href{http://projecteuclid.org/euclid.cmp/1103841481}{\tt{[open access]}}.

\bibitem[Dim82]{Dimock:1982}
J.~Dimock.
\newblock {Dirac quantum fields on a manifold}.
\newblock \href{http://dx.doi.org/10.2307/1998597}{{\em Trans. Amer. Math.
  Soc.} {\bf 269} (1982) no.~1, 133--147}.

\bibitem[Dix77]{Dixmier:1977}
J.~Dixmier.
\newblock {\em {C*-Algebras}}.
\newblock North-Holland-Publishing Company, Amsterdam, New York, Oxford, 1977.

\bibitem[EGBV89]{EstradaGracia-BondiaVarilly:1989}
R.~Estrada, J.~M. Gracia-Bondia, and J.~C. Varilly.
\newblock {On Asymptotic expansions of twisted products}.
\newblock \href{http://dx.doi.org/10.1063/1.528514}{{\em J. Math. Phys.} {\bf
  30} (1989)  2789--2796}.

\bibitem[Ell06]{Ellis:2006}
G.~F.~R. Ellis.
\newblock {The Bianchi models: Then and now}.
\newblock \href{http://dx.doi.org/10.1007/s10714-006-0283-4}{{\em Gen. Rel.
  Grav.} {\bf 38} (2006)  1003--1015}.

\bibitem[Foi83]{Foit:1983}
J.~J. Foit.
\newblock {Abstract Twisted Duality for Quantum Free Fermi Fields}.
\newblock {\em Publ. Res. Inst. Math. Sci. Kyoto} {\bf 19} (1983)  729--741,
  \href{http://www.ems-ph.org/journals/show_abstract.php?issn=0034-5318&vol=19%
&iss=2&rank=11&srch=searchterm|Abstract+Twisted+Duality+for+Quantum+Free+Fermi%
+Fields}{\tt{[open access]}}.

\bibitem[FPH74]{FullingParkerHu:1974}
S.~A. Fulling, L.~Parker, and B.~L. Hu.
\newblock {Conformal energy-momentum tensor in curved spacetime: Adiabatic
  regularization and renormalization}.
\newblock \href{http://dx.doi.org/10.1103/PhysRevD.10.3905}{{\em Phys. Rev.}
  {\bf D10} (1974)  3905--3924}.

\bibitem[FV02]{FewsterVerch:2002}
C.~J. Fewster and R.~Verch.
\newblock {A quantum weak energy inequality for Dirac fields in curved
  spacetime}.
\newblock \href{http://dx.doi.org/10.1007/s002200100584}{{\em Commun. Math.
  Phys.} {\bf 225} (2002)  331--359},
  \href{http://arxiv.org/abs/math-ph/0105027}{\tt{[open access]}}.

\bibitem[Ger68]{Geroch:1968}
R.~P. Geroch.
\newblock {Spinor structure of space-times in general relativity. I}.
\newblock \href{http://dx.doi.org/10.1063/1.1664507}{{\em J. Math. Phys.} {\bf
  9} (1968)  1739--1744}.

\bibitem[Ger70]{Geroch:1970}
R.~P. Geroch.
\newblock {Spinor structure of space-times in general relativity. II}.
\newblock \href{http://dx.doi.org/10.1063/1.1665067}{{\em J. Math. Phys.} {\bf
  11} (1970)  343--348}.

\bibitem[GGBI{\etalchar{+}}04]{GayralGraciaBondiaIochumSchuckerVarilly:2004}
V.~Gayral, J.~M. Gracia-Bondia, B.~Iochum, T.~Schucker, and J.~C. Varilly.
\newblock {Moyal planes are spectral triples}.
\newblock \href{http://dx.doi.org/10.1007/s00220-004-1057-z}{{\em Commun. Math.
  Phys.} {\bf 246} (2004)  569--623},
  \href{http://arxiv.org/abs/hep-th/0307241}{\tt{[open access]}}.

\bibitem[GL07]{GrosseLechner:2007}
H.~Grosse and G.~Lechner.
\newblock {Wedge-Local Quantum Fields and Noncommutative {M}inkowski Space}.
\newblock \href{http://dx.doi.org/10.1088/1126-6708/2007/11/012}{{\em JHEP}
  {\bf 11} (2007)  012}, \href{http://arxiv.org/abs/0706.3992}{\tt{[open
  access]}}.

\bibitem[GL08]{GrosseLechner:2008}
H.~Grosse and G.~Lechner.
\newblock {Noncommutative Deformations of Wightman Quantum Field Theories}.
\newblock \href{http://dx.doi.org/10.1088/1126-6708/2008/09/131}{{\em JHEP}
  {\bf 09} (2008)  131}, \href{http://arxiv.org/abs/0808.3459}{\tt{[open
  access]}}.

\bibitem[GLRV01]{GuidoLongoRobertsVerch:2001}
D.~Guido, R.~Longo, J.~E. Roberts, and R.~Verch.
\newblock {Charged sectors, spin and statistics in quantum field theory on
  curved spacetimes}.
\newblock \href{http://dx.doi.org/10.1142/S0129055X01000557}{{\em Rev. Math.
  Phys.} {\bf 13} (2001)  125--198},
  \href{http://arxiv.org/abs/math-ph/9906019}{\tt{[open access]}}.

\bibitem[GW05]{GrosseWulkenhaar:2005}
H.~Grosse and R.~Wulkenhaar.
\newblock {Renormalisation of phi**4 theory on noncommutative R**4 in the
  matrix base}.
\newblock {\em Commun. Math. Phys.} {\bf 256} (2005)  305--374,
  \href{http://arxiv.org/abs/hep-th/0401128}{\tt{[open access]}}.

\bibitem[Haa96]{Haag:1996}
R.~Haag.
\newblock {\em {Local Quantum Physics - Fields, Particles, Algebras}}.
\newblock Springer, 2 edition, 1996.

\bibitem[Kay85]{Kay:1985}
B.~S. Kay.
\newblock {The double-wedge algebra for quantum fields on Schwarzschild and
  Minkowski space-times}.
\newblock \href{http://dx.doi.org/10.1007/BF01212687}{{\em Commun. Math. Phys.}
  {\bf 100} (1985)  57},
  \href{http://projecteuclid.org/euclid.cmp/1103943337}{\tt{[open access]}}.

\bibitem[Key96]{Keyl:1996}
M.~Keyl.
\newblock {Causal spaces, causal complements and their relations to quantum
  field theory}.
\newblock \href{http://dx.doi.org/10.1142/S0129055X96000093}{{\em Rev. Math.
  Phys.} {\bf 8} (1996)  229--270}.

\bibitem[LR07]{LauridsenRibeiro:2007}
P.~Lauridsen-Ribeiro.
\newblock {\em {Structural and Dynamical Aspects of the AdS/CFT Correspondence:
  a Rigorous Approach}}.
\newblock PhD thesis, Sao Paulo, 2007.
\newblock \href{http://arxiv.org/abs/0712.0401}{\tt{[open access]}}.

\bibitem[LW10]{LongoWitten:2010}
R.~Longo and E.~Witten.
\newblock
\newblock {An Algebraic Construction of Boundary Quantum Field Theory}{\em
  Preprint} (April, 2010)  , \href{http://arxiv.org/abs/1004.0616}{\tt{[open
  access]}}.

\bibitem[MFB92]{MukhanovFeldmanBrandenberger:1990}
V.~F. Mukhanov, H.~A. Feldman, and R.~H. Brandenberger.
\newblock {Theory of cosmological perturbations. Part 1. Classical
  perturbations. Part 2. Quantum theory of perturbations. Part 3. Extensions}.
\newblock \href{http://dx.doi.org/10.1016/0370-1573(92)90044-Z}{{\em Phys.
  Rept.} {\bf 215} (1992)  203--333}.

\bibitem[Mo10]{MorfaMorales:2010}
E. Morfa-Morales. work in progress


\bibitem[O'N83]{O'Neill:1983}
B.~O'Neill.
\newblock {\em {Semi-Riemannian Geometry}}.
\newblock Academic Press, 1983.

\bibitem[OS09]{OhlSchenkel:2009}
T.~Ohl and A.~Schenkel.
\newblock {Algebraic approach to quantum field theory on a class of
  noncommutative curved spacetimes}.
\newblock {\em Preprint} (2009)  ,
  \href{http://arxiv.org/abs/0912.2252}{\tt{[open access]}}.

\bibitem[Ped79]{Pedersen:1979}
G.K. Pedersen.
\newblock {\em {C*-Algebras and their Automorphism Groups}}.
\newblock Academic Press, 1979.

\bibitem[PV04]{PaschkeVerch:2004}
M.~Paschke and R.~Verch.
\newblock {Local covariant quantum field theory over spectral geometries}.
\newblock {\em Class.Quant.Grav.} {\bf 21} (2004)  5299--5316,
  \href{http://arxiv.org/abs/gr-qc/0405057}{\tt{[open access]}}.

\bibitem[Reh00]{Rehren:2000}
K.-H.~Rehren.
\newblock{Algebraic Holography}.
\newblock{\em Annales Henri Poincar\'e} {\bf 1} (2000) 607--623
\href{http://arxiv.org/abs/hep-th/9905179}{\tt{[open access]}}.



\bibitem[Rie92]{Rieffel:1992}
M.~A. Rieffel.
\newblock {\em {Deformation Quantization for Actions of $R^d$}}, volume 106 of
  {\em {Memoirs of the Amerian Mathematical Society}}.
\newblock American Mathematical Society, Providence, Rhode Island, 1992.

\bibitem[San08]{Sanders:2008}
K.~Sanders.
\newblock {\em {Aspects of locally covariant quantum field theory}}.
\newblock PhD thesis, University of York, September, 2008.
\newblock \href{http://arxiv.org/abs/0809.4828}{\tt{[open access]}}.

\bibitem[San09]{Sanders:2009-1}
K.~Sanders.
\newblock {On the Reeh-Schlieder Property in Curved Spacetime}.
\newblock
  \href{http://dx.doi.org/http://dx.doi.org/10%2E1007/s00220-009-0734-3}{{\em
  Commun.Math.Phys.} {\bf 288} (2009)  271--285},
  \href{http://arxiv.org/abs/0801.4676}{\tt{[open access]}}.

\bibitem[San10]{Sanders:2009-2}
K.~Sanders.
\newblock {The locally covariant Dirac field}.
\newblock \href{http://dx.doi.org/10.1142/S0129055X10003990}{{\em Rev. Math.
  Phys.} {\bf 22} (2010) no.~4, 381--430},
  \href{http://arxiv.org/abs/0911.1304}{\tt{[open access]}}.

\bibitem[Sol08]{Soloviev:2008}
M.~A. Soloviev.
\newblock {On the failure of microcausality in noncommutative field theories}.
\newblock \href{http://dx.doi.org/10.1103/PhysRevD.77.125013}{{\em Phys. Rev.}
  {\bf D77} (2008)  125013}, \href{http://arxiv.org/abs/0802.0997}{\tt{[open
  access]}}.

\bibitem[Ste07]{Steinacker:2007}
H.~Steinacker.
\newblock {Emergent Gravity from Noncommutative Gauge Theory}.
\newblock
  \href{http://dx.doi.org/http://dx.doi.org/10%2E1088/1126-6708/2007/12/049}{{%
\em JHEP} {\bf 0712} (2007)  049},
  \href{http://arxiv.org/abs/0708.2426}{\tt{[open access]}}.

\bibitem[Str08]{Strich:2008}
R.~Strich.
\newblock {Passive States for Essential Observers}.
\newblock \href{http://dx.doi.org/10.1063/1.2838155}{{\em J. Math. Phys.} {\bf
  49} (2008) no.~022301, }, \href{http://arxiv.org/abs/0801.3529}{\tt{[open
  access]}}.

\bibitem[Sza03]{Szabo:2003}
R.~J. Szabo.
\newblock {Quantum field theory on noncommutative spaces}.
\newblock {\em Phys. Rept.} {\bf 378} (2003)  207--299,
  \href{http://arxiv.org/abs/hep-th/0109162}{\tt{[open access]}}.

\bibitem[Tay86]{Taylor:1986}
M.~E. Taylor.
\newblock {\em {Noncommutative Harmonic Analysis}}.
\newblock American Mathematical Society, 1986.

\bibitem[TW97]{ThomasWichmann:1997}
L.~J. Thomas and E.~H. Wichmann.
\newblock {On the causal structure of {M}inkowski space-time}.
\newblock \href{http://dx.doi.org/10.1063/1.531954}{{\em J. Math. Phys.} {\bf
  38} (1997)  5044--5086}.

\bibitem[Wal84]{Wald:1984}
R.~M. Wald.
\newblock {\em {General Relativity}}.
\newblock University of Chicago Press, 1984.

\end{thebibliography}
\end{document}